\documentclass[12pt]{article}
\usepackage[T1]{fontenc}
\usepackage{amsfonts,amsmath,amsthm,amssymb,verbatim,bm}
\usepackage[bottom]{footmisc}
\usepackage{chngcntr}
\usepackage[pdftex,letterpaper,left=1.2in,top=1in,bottom=1in,right=1.2in]{geometry}
\usepackage[round,longnamesfirst]{natbib}
\usepackage{graphicx, color}
\usepackage{setspace}
\usepackage{epsf}
\usepackage[margin=10pt,labelfont=bf,labelsep=endash]{caption}
\usepackage{subcaption}
\usepackage{floatrow}
\usepackage{fullpage}
\usepackage[super]{nth}
\usepackage[usenames,dvipsnames]{xcolor}
\usepackage{titlesec}
\usepackage[titletoc]{appendix}
\usepackage[nottoc,notlot,notlof]{tocbibind}
\usepackage[colorlinks,pagebackref=false]{hyperref}
\usepackage{todonotes}
\definecolor{dark-red}{rgb}{0.4,0.15,0.15}
\definecolor{dark-blue}{rgb}{0.15,0.15,0.75}
\definecolor{medium-blue}{rgb}{0,0,0.5}

\def\reducespace{\vspace{-.1in}}

\hypersetup{
    pdftitle={Improving Information from Manipulable Data},
    pdfauthor={Frankel and Kartik},
    citecolor={dark-blue},
    bookmarksnumbered=true,
    urlcolor=BrickRed,
    linkcolor=dark-blue
}

\AtBeginDocument{} %Gets Bib section title to be `References'

\makeatletter
  \renewcommand\@seccntformat[1]{\csname the#1\endcsname.{\hskip.7em\relax}} %Gets period after section title
\makeatother

\newcommand{\appendixref}[1]{\hyperref[#1]{Appendix \ref{#1}}}
\newcommand{\onlineappendixref}[1]{\hyperref[#1]{Supplementary Appendix \ref{#1}}}
\newtheorem{lemma}{Lemma}
\newtheorem{claim}{Claim}
\newtheorem{proposition}{Proposition}

\theoremstyle{remark}
\newtheorem{remark}{Remark}
\newtheorem{fact}{Fact}

\theoremstyle{definition}

\titlespacing\section{0pt}{10pt plus 2pt minus 2pt}{4pt plus 2pt minus 2pt} %Tightens up spacing after section title
\titlespacing\subsection{0pt}{6pt plus 2pt minus 2pt}{2pt plus 2pt minus 2pt} %Tightens up spacing after subsection titile
\titlespacing\subsubsection{0pt}{6pt plus 2pt minus 2pt}{0pt plus 2pt minus 2pt} %Tightens up spacing after subsubsection title

\def\o{\overline}
\def\u{\underline}

\def\E{\mathbb{E}}

\def\Var{\ensuremath{\mathrm{Var}}}
\def\Cov{\ensuremath{\mathrm{Cov}}}
\def\Corr{\ensuremath{\mathrm{Corr}}}

\def\Reals{\mathbb R}
\newcommand{\cov}{\mathrm{Cov}}
\newcommand{\var}{\mathrm{Var}}
\newcommand{\vareta}{\sigma^2_\eta}
\newcommand{\vargamma}{\sigma^2_\gamma}

\newcommand{\loss}{L}
\newcommand{\welfloss}{\mathcal L}
\newcommand{\welfareloss}{\mathrm{Welfare \ Loss}}

\newcommand{\bn}{\beta^{\mathrm{n}}}
\newcommand{\bstar}{\beta^*}
\newcommand{\bstarprime}{\beta^{*\prime}}
\newcommand{\bstarprimeprime}{\beta^{*\prime\prime}}
\newcommand{\hstar}{h^*}
\newcommand{\hstarprime}{h^{*\prime}}
\newcommand{\beq}{\beta^\mathrm{fp}}

\renewcommand{\epsilon}{\varepsilon}
\renewcommand{\phi}{\varphi}
\renewcommand{\bar}{\overline}

\DeclareMathOperator*{\argmin}{arg\,min}

\DeclareMathOperator{\sign}{sign}
\newcommand{\mailto}[1]{\href{mailto:#1}{\texttt{#1}}} 
\def\bi{\begin{itemize}}
\def\ei{\end{itemize}}
\let\oldfootnote\footnote
\renewcommand\footnote[1]{\oldfootnote{\hspace{.4mm}#1}}
\makeatletter
\renewenvironment{proof}[1][\proofname] {\par\pushQED{\qed}\normalfont\topsep6\p@\@plus6\p@\relax\trivlist\item[\hskip\labelsep\bfseries#1\@addpunct{.}]\ignorespaces}{\popQED\endtrivlist\@endpefalse}
\makeatother
\begin{document}
\onehalfspacing

\begin{titlepage}

\title{Improving Information from Manipulable Data\thanks{We thank Ian Ball, Ralph Boleslavsky, Max Farrell, Pepe Montiel Olea, Canice Prendergast, Robert Topel, and various seminar and conference audiences for helpful comments. C\'esar Barilla, Bruno Furtado, and Suneil Parimoo provided excellent research assistance.}}
\author{Alex Frankel\thanks{University of Chicago Booth School of Business; \mailto{afrankel@chicagobooth.edu}.} \and Navin Kartik\thanks{Columbia University, Department of Economics; \mailto{nkartik@columbia.edu}.}}

\maketitle

\thispagestyle{empty}

{

\begin{abstract}
\noindent Data-based decisionmaking must account for the manipulation of data by agents who are aware of how decisions are being made and want to affect their allocations. We study a framework in which, due to such  manipulation, data becomes less informative when decisions depend more strongly on data. We formalize why and how a decisionmaker should commit to underutilizing data. Doing so attenuates information loss and thereby improves allocation accuracy. 
\end{abstract}
}

\bigskip

\quad \emph{JEL Classification:} C72; D40; D82 

\quad \emph{Keywords:} Gaming; Goodhart's Law; Strategic Classification

\bigskip

\end{titlepage}

\begin{comment}
\thispagestyle{empty}

\newpage

\small \singlespacing
\begingroup
\setcounter{tocdepth}{2} % don't print subsubsection in TOC 
\setlength{\parskip}{-1.5pt}
\hypersetup{linkcolor=black}
\tableofcontents
\thispagestyle{empty}
\endgroup
%\tableofcontents 
\end{comment}

\newpage 

\setcounter{page}{1}
\onehalfspacing 
\normalsize
\setlength{\parskip}{6pt plus 1pt minus 1pt}

\section{Introduction}
Firms use a consumer's web browsing history to price discriminate and to target ads. Banks rely on a prospective borrower's credit score to make lending decisions. Search engines take as input a website's text and metadata to produce search results. In these settings and many others, an agent (consumer, borrower, website) generates data that is then used by the designer (firm, bank, search engine) to provide an allocation (prices/ads, interest rates, search rankings). Agents who understand the designer's algorithm can alter their behavior to receive a more desirable allocation. For instance, consumers can adjust browsing behavior to mimic those with low willingness to pay; borrowers can open or close accounts to improve their credit score; and websites can perform search engine optimization. How should a designer maximize allocation accuracy when accounting for the resulting manipulation?

As a benchmark, consider a naive designer who is unaware of the potential for manipulation. Before implementing an allocation rule, the designer gathers data generated by agents and estimates their types (the relevant characteristics). The \emph{naive allocation rule} assigns each agent the allocation that is optimal according to this estimate. But after the rule is implemented, agents' behavior changes: if agents with ``higher observables'' $x$ receive a ``higher allocation'' $y$ under the allocation rule $Y(x)$, and if agents prefer higher allocations, then some agents will find ways to game the rule by increasing their $x$. In line with Goodhart's Law, the original estimation is no longer accurate.\footnote{Goodhart's Law, often rephrased as ``When a measure becomes a target, it ceases to be a good measure,'' was originally stated by \citet{goodhart75} as ``Any observed statistical regularity will tend to collapse once pressure is placed upon it for control purposes.'' In our context ``control purposes'' would correspond to the designer's use of the estimation for allocation decisions.}

A more sophisticated designer realizes that behavior has changed, gathers new data, and re-estimates the relationship between observables and type. After the designer updates the allocation rule based on the new prediction, agent behavior changes once again. The designer might keep adjusting the rule until she reaches a \emph{fixed point}: an allocation rule that is a best response to the data that is generated under this very rule. But the resulting allocation need not match the desired agent characteristics  well.

In this paper we compare the allocation rule chosen by a designer with commitment power---the Stackelberg solution---to the fixed-point allocation rule. We find that in order to improve the accuracy of allocations, a designer should make the allocation rule less sensitive to manipulable data than under the fixed point. In other words, the designer should ``flatten'' the allocation rule. Flattening the allocation results in ex-post suboptimality; the designer has committed to ``underutilizing'' agents' data. Fixed-point allocations, by contrast, are ex-post optimal. However, a flatter allocation rule reduces manipulation, which makes the data more informative about agents' types. Allocation accuracy improves on balance. We develop and explore this logic in what we believe is a compelling model of information loss due to manipulation.

By way of background, note that in some environments, manipulation does not lead to information loss: fixed-point rules deliver the designer's full-information outcome. To see this, think of a fixed-point rule as corresponding to the designer's equilibrium strategy in a signaling game in which the designer and agent best respond to each other. Under a standard single-crossing condition \`a la \citet{Spence73}---the designer wants to give more desirable allocations to agents with higher types, and higher types have lower marginal costs of taking higher observable actions---this signaling game has a fully separating equilibrium, i.e., one in which the designer perfectly matches the agent's allocation to her type. Even with commitment power, a designer cannot improve accuracy by departing from the corresponding allocation rule.

To introduce information loss, we build on a framework first presented by \cite{prendergast1996favoritism}.  The designer learns about an agent's type by observing data the agent generates, her action $x \in \mathbb{R}$. Agents are heterogeneous on two dimensions of their types, what we call  \emph{natural action} and \emph{gaming ability}.  We  initially assume the designer is only interested in the natural action $\eta \in \Reals$, which determines the agent's action $x$  absent any manipulation. Gaming ability $\gamma\in \Reals$ summarizes how much an agent manipulates $x$ in response to incentives. For instance,  in the web search application, $x$ represents all that a search engine sees about a website, $\eta$ the fundamental relevance of a website to a given online query, and $\gamma$ how costly it is for a website's owners to engage in search engine optimization, or how willing they are to do that.

When drawing inferences from the action $x$, the designer's information about the agent's natural action $\eta$ is ``muddled'' with that about gaming ability $\gamma$ \citep{FK19}. We assume the designer observes $x$ and chooses an  allocation $y = Y(x)\in \Reals$ with the goal of minimizing the quadratic distance between $y$ and $\eta$.  We restrict attention to linear allocation rules or policies $Y(x)=\beta x + \beta_0$, and we posit (with microfoundations) that agents adjust their observable $x$ in proportion to $\gamma \beta$---their gaming ability times the sensitivity of allocations to observables.\footnote{As is common, we say ``linear'' instead of the mathematically more precise ``affine''.} These linear functional forms arise in the linear-quadratic signaling models of \citet{FV00} and \citet{BT06}, among others.%\footnote{%\autoref{sec:linassms} microfounds such agent behavior. \autoref{sec:discussnonlinear} discusses the restriction to linear allocation rules. 

%Mathematically, for policies $Y(x) = \beta x +\beta_0$, it is optimal for the designer to depart from the fixed point with $\beta>0$ by attenuating that coefficient towards zero. 
Our main result establishes that the optimal policy under commitment is less sensitive to observables than is the fixed-point policy. 
%Mathematically, for policies of the form $Y(x) = \beta x +\beta_0$, the fixed point satisfies $\beta>0$. We find that it is optimal for the designer to attenuate this coefficient towards zero.
Mathematically, for policies of the form $Y(x) = \beta x +\beta_0$, we find that it is optimal for the designer to attenuate the  fixed point's coefficient $\beta>0$ towards zero. (For this discussion, suppose there is a unique fixed point, for which $\beta>0$; our formal analysis addresses the possibility of multiple or negative fixed points.)
Information is underutilized  at the optimum in the sense that, given the data generated by agents in response to this optimal policy, the designer would ex-post benefit from using a higher $\beta$. For instance, suppose the sensitivity of the naive policy is $\beta=1$: when the designer does not condition the allocation on observables, the linear regression coefficient of type $\eta$ on observable $x$ is 1, and the naive designer responds by matching her allocation rule's sensitivity to this regression coefficient. The fixed-point policy may have $\beta=0.7$. That is, when the designer sets $\beta=0.7$ and runs a linear regression of $\eta $ on $x$ using data generated by the agent in response to $\beta=0.7$,  the regression coefficient is the same $0.7$. Our result is that the optimal policy has $\beta \in (0,0.7)$, say $\beta=0.6$. After the designer sets $\beta=0.6$, however, the corresponding linear regression coefficient is larger than $0.6$, say $0.75$. We emphasize that our argument for shrinking regression coefficients is driven by the informational benefit from reduced manipulation, and in turn, the resulting improvement in allocations. It is orthogonal to concerns about model overfitting.

In comparing our commitment solution with the fixed-point benchmark, it is helpful to keep in mind two distinct interpretations of the fixed point. The first concerns a designer who has market power in the sense that agents adjust their manipulation behavior in response to this designer's policies. Think of websites engaging in search engine optimization 
to specifically improve their {Google} rankings; third party sellers paying for fake reviews on the {Amazon} platform; or citizens trying to game an eligibility rule for a targeted {government} policy. In these cases the designer may settle on a fixed point by adjusting policies until reaching an ex-post optimum. Our paper highlights that this fixed point may yet be suboptimal ex ante, and offers the prescriptive advice of flattening the allocation rule.

A second perspective is that the fixed-point policy represents the outcome of a competitive market.  With many banks, any one bank that uses credit information in an ex-post suboptimal manner will simply be putting itself at a disadvantage to its competitors; similarly for colleges using SAT scores for admissions. So the fixed point becomes a descriptive prediction of the market outcome, i.e., the equilibrium of a signaling game. In that case, our optimal policy suggests a government intervention to improve allocations, or a direction that collusion might take.

%Before turning to the related literature, we stress two points about our approach. First, our paper aims to formalize a precise but ultimately qualitative point, and make salient its logic. Our model is deliberately stylized and, we believe, broadly relevant for many applications. But it is not intended to capture the details on any specific one. We hope that it will be useful for particular applications either as a building block or even simply as a benchmark for thinking about positive and normative implications. Second, we view our main result---the commitment policy flattens fixed points and underutilizes data---as intuitive once one understands the logic of our environment. Indeed, there is a simple first-order gain vs.~second-order loss intuition for a local improvement from flattening a fixed point; see \autoref{lem:lossderivbeq} and the discussion after \autoref{prop:main}. Confirming that the result holds for the global optimum is not straightforward, however; among other complications, there can be multiple fixed points.

Before turning to the related literature, we stress three points about our approach. First, our paper aims to formalize a precise but ultimately qualitative point, and make salient its logic. Our model is deliberately stylized and, we believe, broadly relevant for many applications. But it is not intended to capture the details on any specific one. We hope that it will be useful for particular applications either as a building block or even simply as a benchmark for thinking about positive and normative implications. Second, we view our main result---the commitment policy flattens fixed points and underutilizes data---as intuitive once one understands the logic of our environment. Indeed, there is a simple first-order gain vs.~second-order loss intuition for a local improvement from flattening a fixed point; see \autoref{lem:lossderivbeq} and the discussion after \autoref{prop:main}. Confirming that the result holds for the global optimum is not straightforward, however; among other complications,  the designer's problem is not concave and, separately, there can be multiple fixed points. Third, our main result does, of course, depend on certain important modeling assumptions. We emphasize the result because we find the assumptions compelling. \autoref{sec:extensions} discusses extensions and limitations, including how the result changes under other assumptions.

\paragraph{Related Literature.}

There are many settings in economics in which a designer commits to making ex-post suboptimal allocations in order to improve ex-ante incentives on some dimension. %Our focus is on a designer whose goal is to match an allocation to a state of the world. 
Our specific interest in this paper is in a  canonical problem of matching allocations to unobservables in the presence of strategic manipulation. In this context, we study a simple model in which there is a benefit of  committing to distortions in order to improve the ex-ante accuracy of the allocations.

%\begin{itemize}
%
%\item Compare/contrast with Muddled Information: Mud Inf gives intuition (in a related setting without commitment power by designer) that an agent's actions are less informative about her type when her incentives to manipulate are stronger.  This paper explicitly models the observer's side and gives observer commitment power to improve allocation accuracy. In particular, we explore a basic tradeoff: when an observer more strongly conditions allocations on the data, that strengthens agent's incentives to manipulate, making her action less informative.
%
%
%\item 
%
%
%
%\end{itemize}
Building on the ``linear-quadratic-normal'' signaling games of \citet{FV00} and \citet{BT06}, \citet{FK19} elucidate general conditions under which an agent's action becomes less informative to an observer when the agent has stronger incentives to manipulate.  %These equilibria are the counterparts to fixed points in the current paper. 
%\cite{FK19} elucidate general conditions under which the comparative static holds. 
\citet{FK19} model an observer in reduced form: the agent's payoff is assumed to depend directly on the observer's belief. 
In the current paper, we introduce an explicit accuracy objective for the observer/designer. This allows us to consider commitment power for the designer. We compare the commitment optimum with the fixed point, where fixed points correspond to equilibria in the aforementioned signaling-game papers. The key tradeoff our designer faces is suggested by those papers, and also by \citet{prendergast1996favoritism}: making allocations more responsive to an agent's data amplifies  the agent's manipulation, which makes the data less informative.

%reducing  the optimal responsiveness for allocation accuracy. 

%None of these papers model the allocation-accuracy problem we study here; the latter three papers do not study commitment either. In particular, \citet{FK19} study a reduced-form signaling game in which the agent's payoff depends directly on the observer's belief

%Building on intuitions from \citet{prendergast1996favoritism}, 

%Notwithstanding, our designer faces the following tradeoff suggested by the intuitions in those papers: making allocations more responsive to an agent's data amplifies  the agent's manipulation, which makes the data less informative, reducing  the optimal responsiveness for allocation accuracy. 

Perhaps the most related paper to ours is the contemporaneous work of \citet{ball2020}. He extends the linear-quadratic-elliptical specification  in Section IV of \citet{FK19} to incorporate multiple ``features'' or dimensions; on each feature, agents have heterogeneous natural actions and gaming abilities. His main focus is on optimal scoring rules to improve information, specifically in identifying how to weight the different features when aggregating them into a one-dimensional statistic.\footnote{He interprets the aggregator as produced by an intermediary who shares the decisionmaker's interests, but cannot control the decisionmaker's behavior. That is, the intermediary can commit to the aggregation rule but allocations are made optimally given the aggregation. For work on optimal garbling of signals in other sender-receiver games, see \citet{Whitmeyer20} and references therein.} He also compares his analog of our commitment solution with both his scoring and fixed-point solutions. Similar to   our \autoref{prop:main}, he finds that under certain conditions, his commitment solution is less responsive to all of an agent's features than the (unique, under his assumptions) fixed-point solution.   He does not study the issues tackled by our Propositions \ref{prop:compstats}--\ref{prop:hybrid}.

\cite{BBK20} present a multiple-features model similar to \citet{ball2020}. Like us, they are interested in the commitment solution. Their emphasis, however, is on empirical estimation; they demonstrate their estimator's value using a field experiment.\footnote{\citet{HG20} discuss how to adjust penalized regressions and some other procedures to account for certain kinds of strategic manipulation. The bulk (although not all) of their analysis is in a framework in which manipulation, once properly accounted for, does not entail information loss.}

At a very broad level, our main result that the designer should flatten allocations relative to the fixed-point rule is reminiscent of the ``downward distortion'' of allocations in screening problems following \citet{MR78}. That said, our framework, analysis, and emphasis---on manipulation and information loss, allocation accuracy, contrasting commitment with fixed points---are not readily comparable with that literature. One recent paper on screening to highlight is \cite{bonatti2019consumer}. In a dynamic price discrimination problem, they  show that short-lived firms get better information about long-lived consumers' types---resulting in higher steady-state profits---if a designer reveals a statistic that underweights recent consumer behavior. Suitable underweighting dampens consumer incentives to manipulate demand.
%They study a price discrimination problem in which a designer with access to a long-lived consumer's purchase history chooses what to reveal to a sequence of short-lived firms. %(The design of information to firms who act without commitment can be thought of as yielding policy in between the extremes of fixed point and commitment.)
%They find  that firms get better information about consumer types, and hence higher steady-state profits, if the designer reveals a statistic that underweights recent consumer behavior. Suitable underweighting dampens consumer incentives to manipulate demand.

A finance literature addresses the difficulty of using market activity to learn fundamentals when participants have manipulation incentives. Again in models very different from ours, some papers highlight benefits of committing to underutilizing information. %\footnote{Less directly related, \cite{dworczak2018benchmark} design financial benchmarks to be robust to the incentives of traders to distort these benchmarks; \citepos{zhang2019competition} related work explores the susceptibility of financial derivatives to price manipulation.}
  See, for example, \citet{bond2015government} and \citet{boleslavsky2017selloffs}. These authors study trading in the shadow of a policymaker who may intervene after observing prices or order flows. %The models are fairly disparate in their trading microstructure and the nature of interventions. But a 
The anticipation of intervention makes the financial market less informative about a fundamental %to which the policymaker would like to tailor her intervention. 
to which the intervention should be tailored.
Both papers establish that the policymaker may benefit from a commitment that, in some sense, entails %ex-post
underutilization of information. In particular, \citet[Proposition 2]{bond2015government} highlight a local first-order information benefit vs.~second-order allocation loss akin to our \autoref{lem:lossderivbeq}. Unlike us, they do not study global optimality.

A number of papers in economics study the design of testing regimes and other instruments to improve information extraction. Recent examples include \citet{harbaugh2018coarse} on pooling test outcomes to improve voluntary participation, \citet{perez2017test} on the benefits of noisy tests when agents can manipulate the test, and \citet{MGO19} on using ``conservative'' (or ``confirmatory'') thresholds to mitigate manipulation. \cite{JS20}, \citet{AB20}, and \cite{FK19} analyze how hiding information about agents' actions---increasing privacy---can improve information about their characteristics.\footnote{\citet{EliazSpiegler19} explore the distinct issue of an agent's incentives to reveal her own data to a ``non-Bayesian statistician'' making predictions about her.} %In this vein, \cite{ball2019} proposes that an intermediary should add more noise to the more manipulable dimensions of agents' actions.

Beyond economics, our paper connects to a recent computer science literature studying classification algorithms in the presence of strategic manipulation. See, among others, \cite{HMPW16}, \citet{HIV19}, \cite{milli2018social}, and \citet{kleinberg2019classifiers}. In a binary strategic classification problem, \citet{BravermanGarg19} argue for random allocations to improve allocation accuracy and reduce manipulation costs. %Note that, in contrast to \citet{kleinberg2019classifiers},  we do not model an agent's effort as producing desirable output or affecting the agent's optimal allocation.\footnote{In economics, \cite{prendergast1996favoritism} study a contracting problem in which the principal learns worse payoff-relevant information about the state of the world---a worker's match quality---when incentives for effort are stronger.  \cite{EHM18} study randomized rewards schemes to reduce gaming, and thus improve effort, in a multi-tasking environment. %Their focus is on improving effort rather than information.}  Moreover, our designer is only interested in the accuracy of the allocations, not %(directly) the costs of manipulation.  

We would like to reiterate that our designer is only interested in allocation accuracy, not directly the costs of manipulation. Moreover, unlike \citet{kleinberg2019classifiers}, we model an agent's manipulation effort as pure ``gaming'': it does not provide desirable output or affect the designer's preferred allocation. 
By contrast to us, 
principal-agent problems in economics often focus on how allocation rules interact with incentives for desirable effort. 
For instance, \cite{prendergast1996favoritism} study contracts in which incentivizing worker effort provides a firm worse information about the worker's match quality because of an intermediary's favoritism.  In a multitasking environment, \cite{EHM18} study how randomized rewards schemes can reduce gaming and improve effort. %Their focus is on improving effort rather than information.
\cite{liangmadsen} show that a principal might
strengthen an 
agent's effort incentives by committing to disregard predictive data acquired from other agents; the benefit can dominate the cost of making less accurate predictions.
% which can be worthwhile even when the principal values accurate predictions as well as high effort.
%; the benefit can dominate the value of accurate predictions.

%Finally, we mention the contemporaneous work of \citet{ball2020}. He extends the linear-quadratic-elliptical specification of \citet{FK19} to incorporate multiple ``features'' or dimensions; on each feature, agents have heterogeneous natural actions and gaming abilities. His main focus is on optimal scoring rules to improve information, specifically in identifying how to weight the different features when aggregating them into a one-dimensional statistic.\footnote{He interprets the aggregator as produced by an intermediary who shares the decisionmaker's interests, but cannot control the decisionmaker's behavior. That is, the intermediary can commit to the aggregation rule but allocations are made optimally given the aggregation.} He also compares his analog of our commitment solution with both his scoring and fixed-point solutions. Similar to us, he finds that under certain conditions, his commitment solution is less responsive to all of an agent's features than the (unique, under his assumptions) fixed-point solution.

\section{Model} \label{sec:model}

%\subsection{Agent Types and Designer Welfare}
\subsection{The Environment}
An agent has a type $(\eta, \gamma)\in \mathbb{R}^2$ drawn from  joint distribution $F$. It may be helpful to remember the mnemonics $\eta$ for \emph{n}atural action, and $\gamma$ for \emph{g}aming ability; see \autoref{sec:linassms}. Assume the variances $\var(\eta) = \sigma^2_\eta$ and $\var(\gamma)=\sigma^2_\gamma$ are positive and finite.\footnote{Throughout, we use `positive' %without qualification
 to mean `strictly positive', and similarly for `negative', `larger', and `smaller'.} Denote the means of $\eta$ and $\gamma$ by $\mu_\eta$ and $\mu_\gamma$, respectively, and assume their correlation is $\rho \in (-1,1)$, with $\rho = \cov(\eta, \gamma)/(\sigma_\eta \sigma_\gamma)$.

A designer seeks to match an allocation $y \in \mathbb{R}$ to $\eta$, with a quadratic loss of $
(y - \eta)^2.$
The designer chooses $y=Y(x)$ as a function of an observed action $x \in \mathbb{R}$ that is chosen by the agent.
 %The agent's choice of $x$ will be a function her type $(\eta, \gamma)$ as well as the allocation rule $X(\cdot)$. 
%Since the agent's behavior does not depend directly on $\tau$, we can decompose the designer's payoff as
%\begin{align*}
%W = -\mathbb{E}[(X(a)-\tau)^2] &= -\underbrace{\mathbb{E}[(\te-\tau)^2]}_{\text{Info loss from predicting $\tau$ with $(\eta, \gamma)$}} - \underbrace{\mathbb{E}[(X(a) - \te)^2]}_{\text{Loss beyond what is predictable}}.  
%\end{align*}
%The first term, the information loss from predicting $\tau$ with $(\eta, \gamma)$, is independent of the policy $X$. So without loss of generality, we can take the designer's objective to be the minimization of second term, the loss beyond what is predictable.\footnote{Of course, one could have directly assumed that $\tau = \te$, therefore setting the information loss from predicting $\tau$ from $(\eta, \gamma)$ to zero. Strategically the game would be identical, and notationally things would be cleaner. However, in generalizations we may want to think about adding and subtracting different predictors. In that case it can be useful to keep track of what is possible to learn from a set of predictors versus what cannot be learned.} Call this loss term $\loss$:
Thus, the designer's welfare loss is
\begin{align}
\label{eq:loss}
\welfareloss \equiv \mathbb{E}[(Y(x) - \eta)^2] .
\end{align}
%The designer wants to choose a policy $Y(x)$ to match $y$ to $\te$ as closely as possible. 
The agent chooses $x$ as a function of her type $(\eta,\gamma)$ after observing the allocation rule $Y$. In a manner detailed later, the agent will have an incentive to choose a higher $x$ to obtain a higher $y$.
Given a strategy of the agent, the designer can compute the distribution of $x$ and the value of $\E[\eta|x]$ for any $x$ the agent may choose. A standard decomposition\footnote{\label{fn:decomparg}%To show that \eqref{eq:loss} is equivalent to \eqref{eq:twoparts}, we seek to show that 
The right-hand sides of \eqref{eq:loss} and \eqref{eq:twoparts} are equal if
\[\mathbb{E}[(Y(x))^2 - 2 \eta Y(x) + \eta^2] = \mathbb{E} \big [ \eta^2 - 2 \eta \mathbb{E}[\eta|x] + (\mathbb{E}[\eta|x])^2 + (Y(x))^2 - 2 Y(x) \mathbb{E}[\eta|x] + (\mathbb{E}[\eta|x])^2\big ] .\]
Canceling out like terms and rearranging, it suffices to show that
\[ 2 \mathbb{E}\big [(\mathbb{E}[\eta|x] -  \eta )Y(x) \big ] = 2\mathbb{E} \big [  (\mathbb{E}[\eta|x]-\eta)\mathbb{E}[\eta|x]  \big ]. \]
This  equality holds by the  orthogonality condition $\mathbb{E}[(\mathbb{E}[\eta|x]-\eta) g(x)]=0$ for all functions $g(x)$.} is
\begin{align}
\welfareloss &= \underbrace{\mathbb{E}[(\mathbb{E}[\eta|x]-\eta)^2]}_{\text{Info loss from estimating $\eta$ using $x$}} + \quad \underbrace{\mathbb{E}[(Y(x) - \mathbb{E}[\eta|x])^2]}_{\text{Misallocation loss given estimation}}. \label{eq:twoparts}
\end{align}
%We see that there may be a trade off between the informativeness of the action $x$ with the efficiency of his use of the signal. Using the signal more efficiently by reducing ex post misallocation may lead to a less informative signal ex ante.
Holding fixed the agent's strategy, it is ``ex-post optimal'' for the designer to set $Y(x)=\E[\eta|x]$. However, the agent's strategy responds to $Y$. So the designer may prefer to use an ex-post suboptimal allocation rule to improve her estimation of $\eta$ from $x$, as seen in the first term of \eqref{eq:twoparts}. That is, the designer may benefit from the power to commit to her allocation rule.

\subsection{Linearity Assumptions}
\label{sec:linassms}

%Assume the expectation of $\tau$ given $\eta$ and $\gamma$, $\te$, is linear in $\eta$ and does not depend on $\gamma$: 
%\begin{align}
%\te \equiv \mathbb{E}[\tau|\eta, \gamma] = r \eta + r_0 \label{eq:linpredictable}
%\end{align}
%with $r > 0$.
Assume the designer chooses among {linear allocation rules}: the designer chooses policy  parameters $(\beta, \beta_0)\in \Reals^2$ such that
\begin{align}
Y(x) = \beta x + \beta_0. \label{eq:linpolicy}
\end{align}
Also assume that, given the designer's policy ($\beta,\beta_0$), the agent chooses $x$ using a linear strategy $X_\beta(\eta,\gamma)$ that takes the form
\begin{align}
X_\beta(\eta,\gamma) = \eta + m \beta \gamma \label{eq:linresponse}
\end{align}
for some exogenous parameter $m > 0 $.  Thus $\eta$ is the agent's ``natural action'': the action taken when the designer's policy does not depend on $x$ (i.e., $\beta=0$). The variable $\gamma$ represents idiosyncratic responsiveness to the designer's policy: %agents with higher $\gamma$ increase their actions from their natural levels by more for any $\beta>0$. 
a higher $\gamma$ increases the agent's action from the natural level by more for any $\beta>0$. 
The parameter $m$ captures a common component of responsiveness across all agent types.

One can view the agent's strategy in \autoref{eq:linresponse} as a direct behavioral assumption. But it can be microfounded with a number of agent objectives. To begin with, it is the best response for an agent with $\gamma>0$ who maximizes

%Our preferred interpretation of \autoref{eq:linresponse} is that it is the best response for an agent with $\gamma>0$ who maximizes
\begin{equation}
\label{eq:agentpayoff}
y - (x-\eta)^2/(2m\gamma).\footnote{Substituting in $y=\beta x + \beta_0$, the agent's first-order condition is $\beta -(x-\eta)/(m\gamma)=0$, which implies \autoref{eq:linresponse}. The first-order condition is sufficient because $\gamma>0$.}
\end{equation}
Here we refer to $\gamma$ as an agent's idiosyncratic \emph{gaming ability}. A higher gaming ability scales down the linear marginal cost of taking actions above the natural action. The parameter $m$ captures the ``manipulability'' of the action $x$; a higher $m$ scales down marginal costs for all agent types.

A related interpretation is that the agent chooses a level $b$ by which to boost her ``baseline output'' $\eta$ at cost $b^2/(2m\gamma)$, generating output $x=\eta+b$. The agent knows her cost parameter $\gamma$  when choosing $b$ (and may or may not know $\eta$).\footnote{\citet{LN18} study a signaling model in this vein.}

Alternatively, with the change of variables $e\equiv (x-\eta)/\sqrt{\gamma}$, the agent's payoff \eqref{eq:agentpayoff} can be rewritten as $y-e^2/(2m)$. The setting is thus isomorphic to one in which the agent chooses ``effort'' $e$ at a type-independent cost $e^2/(2m)$ and the designer observes an outcome $x = \eta + e\sqrt{\gamma}$. Here, $\eta$ can be interpreted as the agent's baseline talent while $\gamma$ parameterizes her \emph{marginal product of effort}. When $y=\beta x +\beta_0$, the agent optimally chooses $e=m\beta\sqrt{\gamma}$, as per \autoref{eq:linresponse}.

Finally, the strategy in \autoref{eq:linresponse} can also be motivated as the best response for an agent who maximizes a utility of $$m \gamma y - (x-\eta)^2/2.$$ Here $\gamma \in \Reals$ is an idiosyncratic \emph{marginal benefit} of obtaining a higher allocation $y$. The parameter $m$ captures the ``stakes'' common to all agent types.

\begin{remark}
The signaling specification in Section IV of \citet{FK19}, and those in predecessors cited therein, also model an agent behaving as per \autoref{eq:linresponse}, using the aforementioned microfoundations. They add distributional assumptions on the agent's type that lead to linear allocation rules (or belief updating, in the signaling context), whereas we assume \autoref{eq:linpolicy} directly. We study linear allocation rules for their simplicity, tractability, and comparability; see \autoref{sec:discussnonlinear} for further discussion.
\end{remark}

\subsection{The Designer's Problem}
\label{sec:objective}

The designer commits to her policy $(\beta,\beta_0)$, which the agent observes and responds to according to \eqref{eq:linresponse}. Plugging the rule \eqref{eq:linpolicy} and the strategy \eqref{eq:linresponse} into the welfare loss function \eqref{eq:loss} yields
\begin{align*}\welfareloss&=\mathbb{E}[(\beta (\eta + m \beta \gamma) + \beta_0 - \eta)^2].\end{align*}
The designer's problem is therefore to choose $(\beta, \beta_0)$ to minimize the above loss function, which is quartic in $\beta$.\footnote{Using standard mean-variance decompositions, %algebra shows that 
$$\welfareloss = (1-\beta)^2 \vareta +m^2 \beta^4 \vargamma - 2 (1-\beta) m \beta^2 \rho \sigma_\eta \sigma_\gamma  + (\beta_0 - (1-\beta)\mu_\eta + m \beta^2 \mu_\gamma)^2.$$}
We denote the solution as $(\beta^*,\beta^*_0)$.

\subsection{Discussion} \label{sec:disc}

%\subsubsection{Connecting the model to applications}

%We can now use our model to formally interpret some of the previously mentioned applications.

%Here we review some of the previously mentioned applications, now with the formality of our model. 

%Having seen the model, we now review a few of the previously mentioned applications. 

Let us review some of the applications mentioned earlier using the lens of our model.

In one application, the designer is an internet search platform and the agent is the administrator of a website. The site's true quality---its relevance or value to people searching for certain keywords---is represented by $\eta$. The action or data $x$  is a statistic based on the text and metadata that the platform scrapes off the site. This data can be manipulated through search engine optimization (SEO), with $\gamma$ representing the administrator's skill at or interest in SEO, or alternatively, the resources the administrator has available. The allocation $y$ is the ultimate search ranking, and the platform seeks to rank better-quality sites higher.

Similarly, the designer may be a sales platform and the agent is a third party seller. There, $\eta$ could be the average rating by genuine users of the product (which proxies for product quality), while the statistic $x$ is the observed rating, which can be manipulated by greedy or unscrupulous sellers (those with high $\gamma$) who pay for fake reviews. Here, negative correlation between $\eta$ and $\gamma$ is plausible: less scrupulous sellers may also tend to cut corners on product quality.

A different application is towards testing, in which a college (the designer) evaluates a student (the agent) based on her test score $x$. The allocation $y$ could represent either the priority ranking for admission or the amount of a merit scholarship; in either case, the student values higher $y$. The student's type $\eta$ is her intrinsic aptitude or general high school preparation, while the type $\gamma$ is her skill in or support available for ``studying to the test''. Here, we might expect $\eta$ and $\gamma$ to be positively correlated: both high school preparation and test-taking resources are aided by better socioeconomic status.

Finally,  a firm (the designer) may be allocating a task of importance $y$ to an employee (the agent). The firm seeks to allocate more important tasks to those with more talent, with the employee's talent given  by $\eta$. Recall the ``marginal product of effort'' interpretation of $\gamma$ from \autoref{sec:linassms}: our model can be interpreted as one in which the firm observes the output $x$ of a previous project, with $x= \eta + e \sqrt{\gamma}$ after the employee puts in effort $e$ at cost $e^2/(2m)$.\footnote{Here we abstract away from any direct concerns the firm has for output $x$ or effort $e$. \autoref{sec:addnobjectives} comments on how incorporating those concerns would affect our results.}

At this point we should highlight that our results in \autoref{sec:analysis} depend on the assumption that the designer seeks to match the type dimension $\eta$ rather than $\gamma$. These variables enter asymmetrically into the agent's behavior in \autoref{eq:linresponse}. As highlighted in \cite{FK19}, information about $\eta$ and $\gamma$ can move in opposite directions when the agent is more strongly incentivized to manipulate her action. In the present context, when the designer's policy puts more weight on the data---when $\beta$ increases---the agent's action $x$ becomes less informative about $\eta$; see a formalization in \autoref{rem:lessinfo} below. But the action may simultaneously become more informative about $\gamma$.

In the testing and task allocation applications, a designer could in fact care about $\gamma$ as well as $\eta$. The ability to study for a test, or to increase the  output of a project, could  be correlated with better performance in future classes or on future tasks. \autoref{sec:hybrid} explores how our main result, \autoref{prop:main}, extends if the designer places a limited amount of weight on matching $\gamma$ relative to matching $\eta$, but can flip if she places too much weight on matching $\gamma$.

\subsection{Preliminaries}
\subsubsection{Linear regression of type $\eta$ on action $x$}
\label{sec:regression}

When the designer uses  policy $(\beta, \beta_0)$,  the agent responds with  strategy $X_{\beta}(\eta,\gamma)=\eta+ m \beta \gamma$.  To understand better the designer's welfare loss, suppose the designer were to gather data under that agent behavior and then estimate the relationship between the dimension of interest $\eta$ and the action $x$. Specifically, let $\hat \eta_{\beta}(x)$ denote the best linear estimator of $\eta$ from $x$ under a quadratic loss objective: $$\hat \eta_{\beta}(x)\equiv \hat \beta(\beta) x+\hat \beta_0(\beta),$$
with $\hat \beta$ and $\hat \beta_0$ the coefficients of an ordinary least squares (OLS) regression of $\eta$ on $x$. Following standard results for simple linear regressions, %$\hat \beta$ and $\hat \beta_0$ are the coefficients of an ordinary least squares regression of $\eta$ on $x$. In particular,
\begin{align}
\label{eq:betahat}
\hat \beta(\beta) & = \frac{\Cov(x,\eta)}{\Var(x)},\\ %= \frac{ \vareta + m \rho \sigma_\eta \sigma_\gamma \beta}{\vareta + m^2  \vargamma \beta^2+ 2 m  \rho \sigma_\eta \sigma_\gamma \beta}.
\hat \beta_0(\beta)& =\mu_\eta-\hat \beta(\beta)[\mu_\eta+m\beta\mu_\gamma].\notag
\end{align}
Given the strategy $X_{\beta}$, the covariance of $x$ and $\eta$ is $\Cov(x,\eta)= \vareta + m \rho \sigma_\eta \sigma_\gamma \beta$ and the variance of $x$ is $\Var(x)=\vareta + m^2  \vargamma \beta^2+ 2 m  \rho \sigma_\eta \sigma_\gamma \beta>0$. %Plugging into \eqref{eq:betahat}, we see that $\hat \beta(0) = 1$, and $\hat \beta(\beta) \to 0 $ as $\beta \to \pm \infty$.

It is useful to further rewrite the welfare loss \eqref{eq:twoparts} as follows, for any policy $(\beta,\beta_0)$ defining the linear allocation rule $Y(x)=\beta x+\beta_0$:\footnote{This derivation is identical to that in \autoref{fn:decomparg}, only replacing $\mathbb{E}[\eta|x]$ by $\hat \eta_\beta(x)$ and applying the orthogonality condition $\mathbb{E}[(\hat \eta_\beta(x) - \eta) g(x)]=0$ for all affine functions $g(x)$.}
\begin{align}
\welfareloss%=\mathbb{E}[(Y(x)-\eta)^2] 
%&= \underbrace{ \mathbb{E}[(\mathbb{E}[\eta|x]-\eta)^2]}_{\text{Info loss of estimating $\eta$ from $x$}} + \underbrace{\mathbb{E}[(Y(x) - \mathbb{E}[\eta|x])^2]}_{\text{Misallocation loss given estimation}} \notag
%\\[5pt]
&= \underbrace{\mathbb{E}[(\hat \eta_{\beta}(x)-\eta)^2]}_{\text{Info loss from linearly estimating $\eta$ using $x$}}  + \underbrace{\mathbb{E}[(Y(x) - \hat \eta_{\beta}(x))^2]}_{\text{Misallocation loss given linear estimation}}. \label{eq:twopartslin}
\end{align}
Some readers may find it helpful to note that information loss from estimation (the first term in \eqref{eq:twopartslin}) is the variance of the residuals in an OLS regression of $\eta$ on $x$; put differently, $\mathbb{E}[(\hat \eta_{\beta}(x)-\eta)^2]=\sigma^2_{\eta}(1-R^2_{x\eta})$, with $R^2_{x \eta}$  the coefficient of determination in that regression. %We stress that \autoref{eq:twopartslin} is simply a convenient decomposition; given our focus on linear allocation rules, using OLS entails no restrictions.
\autoref{eq:twopartslin} is a convenient welfare decomposition for linear allocation rules (given the agent's linear response) that is valid because OLS provides the best linear predictor of $\eta$ given $x$.

\begin{remark}
\label{rem:lessinfo}

By a standard property of simple linear regressions, $R^2_{x\eta}$ is the square of the correlation between $x$ and $\eta$: $R^2_{x\eta}=\Corr(x,\eta)^2$. Since $\Corr(x,\eta)=\Cov(x,\eta)/\left(\sigma_{\eta}\sqrt{\Var(x)}\right)$, it is straightforward to confirm using the formulae given earlier for $\Cov(x,\eta)$ and $\Var(x)$ that $\Corr(x,\eta)$ is strictly single peaked in $\beta$ (see \autoref{eq:corr_deriv} in the Appendix), with a maximum of $\Corr(x,\eta)=1$ when $\beta=0$. Furthermore, $\Corr(x,\eta)\geq 0$ when $\rho \geq 0$ and $\beta\geq 0$. Consequently, at least for $\rho\geq 0$ and $\beta\geq 0$, the designer obtains less information about $\eta$ when she chooses a larger $\beta$.\footnote{Less information is not generally in the \citet{Blackwell51} sense unless the prior on $(\eta,\gamma)$ is bivariate normal. Rather, it is in the sense of a higher information loss from linearly estimating $\eta$ using $x$: $\E[(\hat \eta_{\beta}(x)-\eta)^2]$ is increasing in $\beta$.}
\end{remark}

\subsubsection{Benchmark policies}
\label{sec:benchmarks}

\paragraph{Constant.} A rule that does not condition the allocation on the observable corresponds to a constant policy $(\beta, \beta_0)$ with $\beta=0$. A constant policy gives rise to a welfare loss of $\sigma^2_\eta + (\beta_0 - \mu_\eta)^2$. In the decomposition of \autoref{eq:twopartslin}, the entire welfare loss is due to misallocation; the information loss from estimation is zero because the agent's behavior $x=\eta$ fully reveals the natural action $\eta$. Under the constant policy the linear estimator $\hat \eta_0$ has coefficients $\hat \beta(0)=1$ and $\hat \beta_0(0)=0$.

\paragraph{Naive.} If the designer uses a constant policy $(\beta, \beta_0)$ with $\beta=0$, the agent responds with $X_0(\eta,\gamma)=\eta$. Suppose the designer gathers data produced from such behavior, and---failing to account for manipulation---expects the agent to maintain this strategy regardless of the policy. Then the designer would (incorrectly) perceive her optimal policy to be  $(\bn,\bn_0)\equiv (\hat \beta(0), \hat \beta_0(0) )=(1,0)$.  Alternatively, this would be the designer's optimum absent any data manipulation (e.g., were $m=0$ instead of our maintained assumption $m>0$).

\paragraph{Designer's best response.} More generally, suppose the designer expects the agent to use the strategy $X_{\beta}(\eta,\gamma)=\eta+ m  \beta \gamma$ regardless of the designer's policy. %Then, among the available linear allocation rules, the designer finds it optimal to set $$Y(x)=\hat \beta(\beta') x+\hat \beta_0(\beta')$$
The designer would find it optimal in response to set an allocation rule $Y(x)$ equal to the best linear estimator of $\eta$ from $x$, i.e., a policy $(\hat \beta(\beta), \hat \beta_0(\beta))$ yielding $Y(x) = \hat \eta_{\beta} (x)$.

\paragraph{Fixed point.} We say that a policy $(\beq, \beq_0)$ is a \emph{fixed point} if $$\beq=\hat \beta(\beq) \quad \text{ and } \quad \beq_0 = \hat \beta_0(\beq).$$
%A fixed point corresponds to a Nash equilibrium of a game in which the designer's policy (chosen among linear policies) is set simultaneously with the agent's strategy.
Fixed points correspond to the  pure-strategy Nash equilibria of a game in which the designer's policy (chosen among linear policies) is set simultaneously with the agent's strategy (with the agent's best response given by \autoref{eq:linresponse}).  
This simultaneous-move game would eliminate the designer's commitment power. That is, instead of the designer committing to a policy---the Stackelberg solution---the policy is a best response to the agent's strategy that the policy induces. Under a fixed-point policy the designer uses information ex-post optimally: in the decomposition of \autoref{eq:twopartslin}, a fixed-point policy has zero misallocation loss.

\begin{figure}

\begin{center}

\begin{subfigure}{.9 \textwidth}

\begin{center}
\includegraphics[width=4 in]{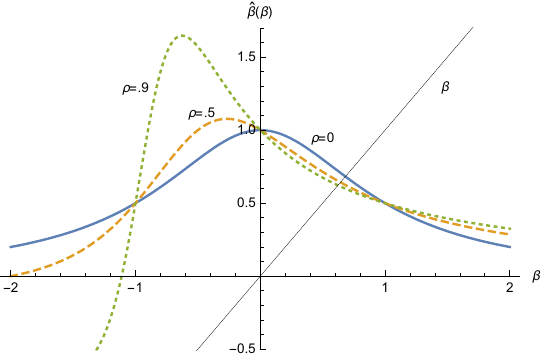}

\end{center}

\caption{\label{panel:posrho}Parameters: $\sigma_{\eta}=\sigma_{\gamma}=1$ and $m=1$.}

\end{subfigure}

\vskip .5in

\begin{subfigure}{.9 \textwidth}

\begin{center}
\includegraphics[width=4 in]{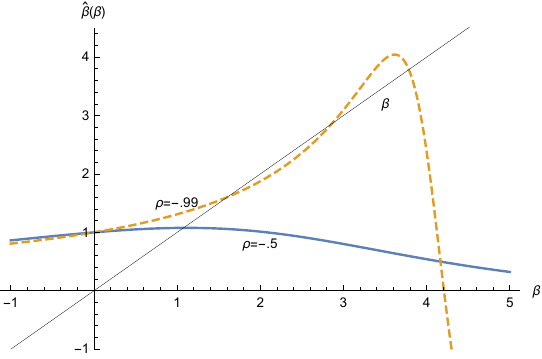}
\end{center}

\caption{\label{panel:negrho}Parameters: $\sigma_{\eta}=\sigma_{\gamma}=1$ and $m=0.24$.}

\end{subfigure}

\caption{\label{fig:betahat} The best response function $\hat \beta$. Intersections of $\hat \beta$ with $\beta$ correspond to fixed-point sensitivities $\beq$. \autoref{panel:posrho} illustrates that when $\rho\geq 0$, $\hat \beta$ is decreasing on $[0,\infty)$ and hence there is a unique positive fixed point. \autoref{panel:negrho} illustrates that there can be multiple positive fixed points when $\rho<0$.}

\end{center}

\end{figure}

\autoref{fig:betahat} illustrates some designer best response functions and fixed points. There can, in general, be multiple fixed points, including ones with negative sensitivity or weight on the agent's action (i.e., $\beq<0$). However, we are interested in, and will focus on, fixed points with positive sensitivity---positive fixed points, for brevity. The following result justifies our focus.

\begin{proposition}
\label{lem:pos_beq}
Any fixed point has $\beq\neq 0$.
There exists a fixed point with $\beq>0$. If $\rho\geq 0$, there is a unique positive fixed point, and it satisfies $\beq \in (0,1)$.
\end{proposition}

The proof uses a routine analysis of \autoref{eq:betahat}. A sensitivity of $\beta=0$ is not a fixed point because $\hat \beta(0)=1$.  A positive fixed point exists because $\hat \beta(\cdot)$ is continuous and $\hat \beta(\beta) \to 0 $ as $\beta \to \infty$, which reflects that the agent's action is uninformative about $\eta$ at the limit. The uniqueness result is because nonnegative correlation, $\rho\geq 0$, implies $\hat \beta(\cdot)$ is strictly decreasing on $[0,\infty)$, as seen in \autoref{panel:posrho}. \autoref{panel:negrho} illustrates that there can be multiple positive fixed points when $\rho<0$. That there is only one positive fixed point when $\rho\geq 0$ has been noted in different form in \citet[Proposition 4]{FK19}.

\autoref{lem:pos_beq} implies that any fixed point has positive information loss: the first term in \autoref{eq:twopartslin} is positive whenever $\beta\neq 0$. The information loss owes to our maintained assumption that $\sigma_\gamma>0$; were $\sigma_\gamma=0$, instead, the agent's strategy from \autoref{eq:linresponse} would fully reveal $\eta$ no matter a policy's sensitivity $\beta$.

\section{Analysis} \label{sec:analysis}

\subsection{Main Result}

We seek to compare the designer's optimal policy $(\beta^*,\beta^*_0)$  with the fixed points $(\beq,\beq_0)$. Take any fixed-point sensitivity $\beq>0$. Our main result,  \autoref{prop:main} below, is that the optimal policy puts less weight on the agent's action than does the fixed point. Furthermore, the optimal policy underutilizes information by putting less weight on the agent's action than does the OLS coefficient (and hence the best linear policy) given the data generated by the agent in response.

\begin{proposition}	
\label{prop:main}
There is a unique optimum, $(\beta^*,\beta^*_0)$. It has $\bstar >0$ and $\bstar<\beq$ for any fixed point with $\beq>0$. Moreover, $\hat \beta(\bstar)>\bstar$.
\end{proposition}

For a concrete example, take $m=\vareta=\vargamma=1$ and $\rho=0$. Recall that the sensitivity of the naive policy is (normalized to) $\beta=1$. The unique fixed-point policy has $\beq\approx0.68$. The optimal policy reduces the sensitivity to $\bstar\approx 0.59$.  Given the agent's behavior under this policy, the designer would ex post prefer the higher value $\hat \beta(\bstar)\approx 0.74$. In this example not only is $\hat \beta(\bstar)>\bstar$, but  $\hat \beta(\bstar)>\beq$; we explain subsequently that this point holds whenever the correlation $\rho$ is nonnegative.
%For %$k=.955417$ and $\rho=-.49$, we have $\beq=.7000$, $\bstar=.6000$, and $\hat \beta[\bstar]=.7638$.
%To reiterate, we see that the designer -- an online platform ranking web sites or third party sellers, a set of colleges admitting students, a firm allocating tasks to employees -- should commit to making its allocation rule less sensitive to manipulable data than is ex post optimal.

 %: when the designer does not condition the allocation on observables, the linear regression coefficient of type $\eta$ on observable $x$ is 1, and the naive designer responds by matching her allocation rule's sensitivity to this regression coefficient. 
%The fixed-point policy may have $\beta=0.7$. That is, when the designer sets $\beta=0.7$, %and runs a linear regression of $\eta $ on $x$ (using data generated by the agent in response to $\beta=0.7$),  the data subsequently generated by the agent leads to the same regression coefficient of $0.7$. 
%Our result is that the optimal policy has $\beta \in (0,0.7)$, say $\beta=0.6$. Note that the designer recognizes and commits to ex-post misallocations under this optimal policy: after the designer sets $\beta=0.6$, the corresponding linear regression coefficient could be $\simeq 0.75$. %For $k=1$ and $\rho=0$, we have $\beq=.682$, $\bstar=.589$, and $\hat \beta(\bstar)=.742$. For $k=.955417$ and $\rho=-.49$, we have $\beq=.7000$, $\bstar=.6000$, and $\hat \beta[\bstar]=.7638$.

Here is the intuition for the comparison of the optimum with fixed points, as illustrated graphically in \autoref{fig:losses}. Consider a designer choosing $\beta=\beq>0$. When  paired with the correspondingly optimal $\beta_0$, this policy is  ex-post optimal in the sense that misallocation loss (the second term in the welfare decomposition \eqref{eq:twopartslin}) given the information the designer obtains about $\eta$ is minimized at zero.  Adjusting the sensitivity $\beta$ in either direction from $\beq$ increases  misallocation loss, but this harm is second order because we are starting from a minimum.  By contrast, at $\beta=\beq$ there is positive information loss from estimation (the first term in \eqref{eq:twopartslin}) because the agent's action does not reveal $\eta$. Lowering $\beta$ reduces information loss from estimation, which yields a first-order benefit. (The first-order benefit was suggested by \autoref{rem:lessinfo} for $\rho \geq 0$, and the point is general, as elaborated below.) Hence, there is a net first-order welfare benefit of lowering $\beta$ from $\beq$. Of course, the designer wouldn't lower $\beta$ down to 0, since making some use of the information from data is better than not using it at all.\footnote{Indeed, any fixed-point policy itself does better than the best constant policy $(\beta,\beta_0)=(0,\mu_\eta)$. Note, however, that this constant policy can be better than the naive policy $(\bn,\bn_0)=(1,0)$.} 
%COULD SHOW CONSTANT BEING BETTER THAN NAIVE IN THE FIGURE

The proof of \autoref{prop:main} in \autoref{app:proofs} establishes uniqueness of the global optimum, rules out that it is negative, and shows that it is less than every fixed point with $\beq>0$. \autoref{lem:lossderivbeq} formalizes a key step, the aforementioned first-order benefit of reducing $\beta$ from any $\beq$. To state the lemma, let $\welfloss(\beta)$ be the welfare loss from policy $\beta$ (paired with the correspondingly optimal $\beta_0$), with derivative $\welfloss'(\beta)$.%\footnote{We write $\welfloss(\beta)$ rather than $\welfloss(\beta,\beta_0)$ because for any $\beta$ there is a uniquely optimal $\beta_0$ that can be substituted in; see the proof of \autoref{prop:main}, which also confirms that $\welfloss(\cdot)$ is differentiable.} 

\begin{lemma}
\label{lem:lossderivbeq}
%Consider any $\beq$ satisfying $\hat \beta(\beq) = \beq$. 
For any $\beq$, it holds that $\welfloss'(\beq) > 0$.
\end{lemma}

\autoref{lem:lossderivbeq} applies regardless of the sign of the correlation parameter $\rho$ and also applies to negative values of $\beq$ when those exist. Here is the logic for the lemma. As noted above, starting from a fixed point the first-order change in welfare loss is just the change in information loss from estimation. Recall from \autoref{sec:regression} that information loss is proportional to $1-R^2_{x \eta}$, and $R^2_{x \eta}=\Corr(x,\eta)^2$. By \autoref{eq:betahat}, the sign of $\Corr(x,\eta)$ is that of $\hat \beta$, which at a fixed point is the same as the sign of $\beq$. \autoref{rem:lessinfo} established that, regardless of the sign of $\rho$, $\Corr(x,\eta)$ is increasing in $\beta$ for $\beta < 0$ and is decreasing in $\beta$ for $\beta>0$. Putting these facts together, at a positive fixed point $\Corr(x,\eta)$ is positive and locally  decreasing in $\beta$, while at a negative fixed point $\Corr(x,\eta)$ is negative and locally increasing in $\beta$. In either case, information loss, which scales with $1-\Corr(x,\eta)^2$, is locally increasing in $\beta$.

\begin{figure}
\begin{center}
\includegraphics[width=6in]{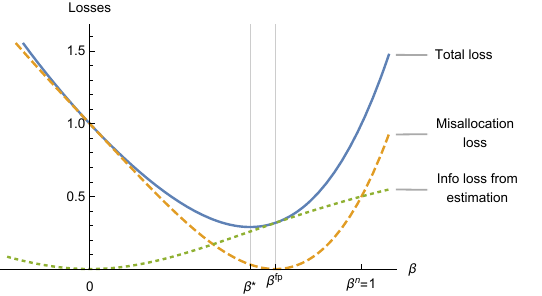}
\end{center}

\caption{\label{fig:losses} The welfare loss decomposition from \autoref{eq:twopartslin} for policy $(\beta_0, \beta)$, with the optimal $\beta_0$ plugged in for each $\beta$ on the horizontal axis. Parameters: $\sigma_{\eta}=\sigma_{\gamma}=m=1$ and $\rho=0$. Numerical solutions: $\bstar=0.590$ and $\beq=0.682$.}

\end{figure}

The last part of \autoref{prop:main} says that information is underutilized at the optimum: $\hat\beta (\bstar)>\bstar$. The logic for this result can be readily understood using \autoref{fig:betahat}. As seen there, $\hat \beta(0)>0$, and hence by continuity, $\hat \beta(\beta)>\beta$ for all positive $\beta$ less than the smallest positive fixed point.\footnote{Note that the underutilization result requires $\bstar$ to be less than all positive fixed points, not just some of them. For example, in the $\rho=-0.99$ curve in \autoref{panel:negrho}, $\hat \beta(\beta)<\beta$ for all $\beta$ in between the smallest and the middle fixed points.} Moreover, as seen in  \autoref{panel:posrho}, $\rho\geq 0$ implies that $\hat \beta(\cdot)$ is strictly decreasing on $[0,\infty)$, and hence $\hat \beta(\bstar)>\beq$ for the unique positive fixed point $\beq$. But the $\rho=-.99$ curve in \autoref{panel:negrho} shows that negative correlation can lead to $\hat \beta(\bstar)<\beq$ for all fixed points $\beq$.

%Turning to \autoref{prop:main}'s result on underutilizing information, $\hat\beta (\bstar)>\bstar$ holds because $\hat \beta(\beta)>\beta$ for all positive $\beta$ less than the smallest positive fixed point: as illustrated in \autoref{fig:betahat}, the function $\hat \beta(\cdot)$ is continuous, $\hat \beta(0)>0$, and $\hat \beta(\beq)=\beq$ for any $\beq$.   \autoref{panel:negrho} highlights that our underutilization conclusion requires $\bstar$ to be less than every positive fixed point, not just some of them: in the $\rho=-0.99$ curve, we see that for $\beta$ in between the smallest and the middle fixed points, $\hat \beta(\beta)<\beta$. Furthermore, since in that $\rho=-.99$ example  $\hat \beta(\cdot)$ is strictly increasing on $[0,\beq]$ for the smallest fixed point $\beq$, $\hat \beta(\bstar)\in(\bstar,\beq)$ for the smallest---and hence any---$\beq$ there. By contrast, $\rho\geq 0$ assures that $\hat \beta(\bstar)>\beq$ for the unique positive fixed point $\beq$. The reason is that $\rho\geq 0$ implies $\hat \beta(\cdot)$ is strictly decreasing on $[0,\infty)$, as seen in \autoref{panel:posrho}.

\begin{remark}
The welfare gains from commitment can be substantial. As  $\rho \to -1$ and for suitable other parameters (viz., $m \sigma_\gamma/\sigma_\eta \to 1/4^+$) the unique fixed point's welfare is arbitrarily close to that of the best constant policy \mbox{$Y(x)=\mu_\eta$}, while the optimal policy's welfare is arbitrarily close to the first best's. The welfare of both the first-best policy and the constant policy are independent of $\rho$; the former is $0$ (our normalization) while the latter is $-\sigma^2_\eta$, which can be arbitrarily low.
\end{remark}

\begin{remark}
\label{rem:bstar<1}
When correlation $\rho$ is nonnegative, the optimal sensitivity $\bstar$ is less than a naive designer's choice of $\beta=1$ (see \autoref{sec:benchmarks}). This follows from the unique positive fixed point satisfying $\beq<1$ when $\rho\geq 0$ (\autoref{lem:pos_beq}) and $\bstar \in (0, \beq)$ (\autoref{prop:main}). However, when $\rho<0$ it is possible that $\bstar>1$. In fact, the proof of \autoref{prop:main} yields a characterization: $\bstar <1$ if $2m\sigma_\gamma / \sigma_\eta>-\rho$; $\bstar>1$ if $2 m\sigma_\gamma / \sigma_\eta<-\rho$; and $\bstar=1$ if $2 m\sigma_\gamma / \sigma_\eta=-\rho$. See \autoref{claim:bstarnaive} in \appendixref{appsec:proofmain}.
\end{remark}

\subsection{Comparative Statics}

We provide a few comparative statics below. In taking comparative statics, it is helpful to observe that the designer's best response $\hat \beta(\beta)$ defined in \autoref{eq:betahat} depends on parameters $m$, $\sigma_\eta$, and $\sigma_\gamma$ only through the statistic $k\equiv m\sigma_\gamma / \sigma_\eta$, as does the welfare loss $\welfloss(\beta)$ divided by $\sigma^2_\eta$ (see \autoref{eq:L} in \appendixref{appsec:proofmain}). 
Therefore, the optimal and fixed-point values $\bstar$ and $\beq$ also only depend on these parameters through $k$. The statistic $k$ summarizes the susceptibility of the allocation problem to manipulation: higher $k$ (arising from higher stakes or manipulability $m$ of the mechanism, greater variance in gaming ability $\sigma^2_\gamma$, or lower variance in natural actions $\sigma^2_\eta$) means that under any given policy, agents  %as a whole
adjust their observable action $x$ further from their natural action $\eta$, relative to the spread of observables prior to manipulation. Hence, for comparative statics of $\bstar$ and $\beq$ over model primitives, it is sufficient to consider only the statistic $k$ and the correlation parameter $\rho$.

{\samepage
\begin{proposition}
\label{prop:compstats}
For $k\equiv m\sigma_\gamma / \sigma_\eta$, the following comparative statics hold for $\bstar$ and $\bstar/\beq$.\reducespace
\begin{enumerate}
	\item \label{compstats1} As $k \to \infty$, $\bstar \to 0$; as $k\to 0$, $\bstar \to 1$. If $\rho\geq 0$, then $\beta^*$ is strictly decreasing in $k$; if $\rho <0$, then $\beta^*$ is strictly quasiconcave in $k$, attaining a maximum at some point.
	\item \label{compstats2} $\beta^*$ is strictly increasing in $\rho$ when $k>3/4$, strictly decreasing in $\rho$ when $k<3/4$, and independent of $\rho$ when $k=3/4$.
	\item \label{compstats3}When $\rho=0$, $\bstar/\beq$ is strictly decreasing in $k$, approaching $\sqrt[3]{1/2}\approx 0.79$ as $k\to \infty$ and $1$ as $k\to 0$. 	
\end{enumerate}
\end{proposition}
}

The limits in part \ref{compstats1} of the proposition are intuitive. From \autoref{eq:linpolicy} and \autoref{eq:linresponse}, any particular $\beta>0$ would result in an arbitrarily large allocation for any given type as the manipulability parameter $m\to \infty$ (or, more generally, the statistic $k \to \infty$). Hence, the optimum $\bstar$ must go to $0$ at this limit. As $k\to 0$, by contrast, there is no manipulation, and hence the naive policy becomes optimal: $\bstar \to 1$. Furthermore, when correlation is nonnegative, $\rho\geq 0$, a designer faced with a more manipulable environment (larger $k$) should put less weight on the agent's action; the intuition is simply that the agent's action becomes less informative. However, when $\rho<0$, an increase in $k$ can actually make the agent's action more informative for a given $\beta$. That the optimum $\bstar$ is no longer monotonically decreasing in $k$ follows from the limit $\bstar \to 1$ as $k\to 0$  and \autoref{rem:bstar<1}'s observation that for any $\rho<0$, $\bstar>1$ when $k$ is sufficiently small.

Turning to part \ref{compstats2} of the proposition, one might expect greater correlation to increase the optimum $\bstar$, at least when correlation is nonnegative. But this turns out to hold only when the susceptibility-to-manipulation statistic $k$ is large enough. Here is an explanation. The formal proof shows that the cross partial derivative of welfare loss with respect to $\beta$ and $\rho$ is %$2 k \beta (3\beta-2)$, whose sign is 
positive for %$\beta>2/3$ and negative when $\beta<2/3$. 
$\beta$ above a threshold between $0$ and $1$ and negative for $\beta$ below that threshold. When $k$ is small, the designer can restrict attention to values of $\beta$ close to $1$, and hence the optimum is decreasing in $\rho$; when $k$ is large, it is $\beta$ close to $0$ that is relevant, and hence the optimum is increasing in $\rho$.

Finally, part \ref{compstats3} of \autoref{prop:compstats} implies that when the agent's characteristics are uncorrelated, the ratio  $\bstar/\beq$ decreases as the statistic $k$ increases.  As $k\to 0$, the fixed point fully reveals an agent's natural action ($\beq\to 1$) and so the designer does not benefit from commitment power: the fixed point is optimal as it provides the minimum possible welfare loss. As $k\to \infty$, both $\bstar$ and $\beq$ tend to zero yet the ratio $\bstar/\beq$ stays bounded. 

We also have the following comparative statics in welfare:

\begin{proposition}
\label{prop:compstatswelfare}
The designer's welfare loss at the optimum, $\welfloss(\bstar)$, is strictly increasing in $\sigma_\gamma$ and $m$; it is also strictly increasing in $\sigma_\eta$ when $\rho\geq 0$, but for $\rho<0$ it is strictly quasiconvex, attaining a minimum at $\sigma_\eta=m\sigma_\gamma/(-2\rho)$. Finally, for $\rho\geq 0$, the welfare loss at the optimum is strictly decreasing in $\rho$.
\end{proposition}

Since the agent's action is $x=\eta+m\beta \gamma$, it is intuitive that an increase in either $\sigma_\gamma$ or $m$ makes actions less informative about $\eta$, and hence reduces the designer's welfare. Indeed, $\welfloss(\bstar)$ divided by $\vareta$, which is $1-R^2_{x \eta}$ (\autoref{sec:regression}),  depends only on $\rho$ and $k\equiv m \sigma_\gamma/\sigma_\eta$, and is increasing in $k$; see \autoref{lem:dLdk} in \autoref{app:proofs}. The %welfare 
effect of an increase in $\sigma_\eta$ is more nuanced. %Since $\welfloss(\bstar)=\vareta \times \welfloss(\bstar)/\vareta$, we see that there are competing effects: 
%the increase in the first term $(\vareta)$ reflects higher welfare loss from more baseline uncertainty about $\eta$, whereas the second term ($\welfloss(\bstar)/\vareta$) decreases because spreading out natural actions mutes the noise from heterogenous gaming ability.
%Expanding out 
 Writing
welfare loss $\welfloss(\bstar)$ as 
$\vareta \times \welfloss(\bstar)/\vareta$, we see that increasing $\sigma_\eta$ has competing effects:
the first term $(\vareta)$ increases due to more baseline uncertainty about $\eta$, whereas the second term ($\welfloss(\bstar)/\vareta$) decreases because spreading out natural actions mutes the noise from heterogenous gaming ability.
 The first effect unambiguously dominates for $\rho\geq 0$, but it turns out that for any $\rho<0$, if (and only if) $\sigma_\eta$ is sufficiently small then the second effect dominates and welfare loss %$\welfloss (\bstar)$ 
is decreasing in $\sigma_\eta$.\footnote{Another quantity of interest is $\vareta-\welfloss(\bstar)$, the welfare gain from the optimal policy over the best constant policy. Regardless of $\rho$, this quantity is increasing in $\sigma_\eta$ because $\vareta-\welfloss(\bstar)=\vareta\left(1-\welfloss(\bstar)/\vareta\right)$ and $\welfloss(\bstar)/\vareta$ is decreasing in $\sigma_\eta$.}  Finally, the intuition for the welfare loss $\welfloss (\bstar)$ decreasing in $\rho$ when $\rho\geq 0$ is that a greater nonnegative correlation is akin to reducing the heterogeneity in gaming ability, which improves information.\footnote{\citet[Proposition 4]{FK19}
 note the same comparative statics in $\sigma_\gamma$, $\rho$, and (their analog of) $m$ for the unique positive fixed point when $\rho\geq 0$.}

%An increase in $\sigma_\eta$ is more nuanced. On the one hand welfare loss increases because there is more baseline uncertainty about $\eta$. But on the other hand, for any choice of $\beta$, spreading out natural actions mutes the noise from heterogenous gaming ability is muted; as noted before \autoref{prop:compstats}, the welfare loss $\welfloss(\beta)$ divided by $\sigma^2_\eta$ depends only on $m \sigma_\gamma/\sigma_\eta$.

\section{Discussion}
\label{sec:extensions}

\subsection{Mixed Dimensions of Interest}
\label{sec:hybrid}
As mentioned in \autoref{sec:disc}, in some settings a designer may care about matching the allocation to not just the agent's natural action $\eta$ but also the gaming ability $\gamma$. For instance, both dimensions may be predictive of future performance in school or at job tasks.  Accordingly, we consider in this subsection (alone) a designer whose welfare loss is given by
\begin{equation}
\label{eq:hybridobj}	
%\welfareloss \equiv 
\mathbb{E}[(Y(x) - \left[(1-\kappa) \eta+ \kappa \gamma \right])^2],
\end{equation}
for some exogenous parameter $\kappa \in (0,1)$.\footnote{It is equivalent to posit a convex combination of quadratic losses from mismatching $\eta$ and $\gamma$ rather than a quadratic loss from mismatching the convex combination of $\eta$ and $\gamma$. That is, the designer's objective could equivalently be
$\E[(1-\kappa)(Y(x) - \eta)^2 + \kappa (Y(x)-\gamma)^2].$} 
%MAYBE WE WANT TO INCLUDE kappa=0 TO SUBSUME THE BASELINE; ALSO CONSIDER WHETHER WE WANT TO INCLUDE 1.
%For one motivation, consider the marginal product of effort interpretation discussed in \autoref{sec:linassms}. By exerting effort $e$, a worker generates output $x=\eta+e\sqrt{\gamma}$ at cost $e^2/(2m)$.  The organization wants to assign new tasks to the worker in line with her future productivity, which could depend on a combination of her baseline talent type $\eta$ and her marginal product of effort type $\gamma$. 
The parameter $\kappa$ reflects the relative importance of gaming ability $\gamma$, compared to a $1-\kappa$ weight on natural action $\eta$. % two characteristics for future productivity. 

When considering objective \eqref{eq:hybridobj}, we continue to assume the designer uses linear policies of the form $Y(x)=\beta x+\beta_0$ and the agent responds according to  $X_{\beta}(\eta,\gamma)=\eta+m\beta\gamma$. The optimal policy now minimizes \eqref{eq:hybridobj}. Just as in \autoref{sec:regression}, given any $X_\beta$ the designer can calculate $\left(\hat \beta(\beta),\hat \beta_0(\beta)\right)$ as the OLS regression coefficients of $(1-\kappa)\eta+\kappa \gamma$ on $x$. A fixed-point policy $(\beq,\beq_0)$ is one in which $\beq=\hat \beta(\beq)$ and $\beq_0=\hat \beta_0(\beq)$.

The key intuition underlying our main result---when the designer cares only about matching $\eta$, it is optimal to reduce allocation sensitivity from a fixed point---is that decreasing manipulation incentives improves information about $\eta$ (\autoref{lem:lossderivbeq}). If the designer cared instead only about matching $\gamma$, the logic from \citet{FK19} suggests the opposite should hold. Intuitively, when the allocation sensitivity $\beta>0$ is larger,  the  variation in the observable $x$ $(=\eta+m\beta \gamma)$ depends more on $\gamma$ and less on $\eta$; hence, increasing manipulation incentives \emph{increases} information about $\gamma$. Put differently, when the designer cares only about matching $\gamma$, variation in $\eta$ simply adds noise to $x$; increasing $\beta$ effectively scales down that noise.

More generally, one might expect that a designer who puts sufficient weight on matching $\eta$ would optimally reduce the sensitivity from a fixed point, while a designer who weights matching $\gamma$ sufficiently would increase the sensitivity. We can establish such a result cleanly under the simplifying assumption that $\eta$ and $\gamma$ are uncorrelated, i.e., $\rho=0$.

\begin{proposition}
\label{prop:hybrid}
Assume $\rho=0$ and designer welfare loss \eqref{eq:hybridobj}. %Among fixed points with positive sensitivity, there is unique one, denoted $(\beq,\beq_0)$ with $\beq>0$. 
%There is a unique fixed point with positive sensitivity, denoted $(\beq,\beq_0)$.\footnote{I.e., we do not rule out fixed points with negative sensitivity.}
There is a unique fixed point among those with positive sensitivity, which we denote $(\beq,\beq_0)$.
Let \mbox{$\bar \kappa \equiv 1-\frac{\sqrt{1+4m}-1}{2m}$} $\in (0,1)$. It holds for the unique optimum $(\bstar,\bstar_0)$ that:
\begin{align*}
\kappa<\bar \kappa &\implies \bstar \in (0, \beq),\\
\kappa = \bar \kappa &\implies \bstar = \beq,\\
\kappa>\bar \kappa &\implies \bstar > \beq.
% IS THERE AN UPPER BOUND ON \bstar IN THE LAST CASE?	
\end{align*}
\end{proposition}

For $\rho=0$, the proposition explicitly identifies identifies a critical threshold, $\bar \kappa \in (0,1)$, such that the fixed point with $\beq>0$ is (only) optimal when the designer's weight on matching $\gamma$ is $\kappa=\bar \kappa$. When $\kappa < \bar \kappa$, it is optimal to flatten the fixed point; when $\kappa > \bar \kappa$, it is optimal to steepen. These conclusions extend and qualify \autoref{prop:main}. Interestingly, the threshold weight $\bar \kappa$ is increasing in the manipulation parameter $m$, but does not depend on the variances $\vareta$ and $\vargamma$.

%Give some interpretation/discussion of result.

%Maybe include graph of $\beq$ and $\bstar$ against $\kappa$.

\subsection{Additional Designer Objectives}
\label{sec:addnobjectives}
Information about the agent's type might only affect part of the designer's welfare.  Our maintained assumption is that agents shift their action by $ m \beta \gamma$ away from their natural action. If the designer seeks to induce higher actions, then the designer will commit to increase $\beta$ above the optimum that was based on allocation accuracy.  This could occur in the task allocation application where the action $x$ corresponds to the output of an evaluation task---this output may be directly valuable to the firm. The designer would have the same incentives in the school testing application if she didn't care per se about a student's test score $x$, but cared about inducing the student to put in effort ($e = (x-\eta)/\sqrt{\gamma}$) to study for the exam.

 If the designer instead wants to reduce the agent's distortions, or the designer internalizes the costs of distortion, she will  weaken manipulation incentives by attenuating $\beta$ towards zero. This would be relevant to a government choosing the eligibility rule for a targeted policy. A government that values citizens' welfare is harmed by the costly effort devoted to gaming eligibility. 

Under either of the additional objectives mentioned in the two previous paragraphs, the designer's ex-post preference---given the agent's action---remains to match the allocation to the agent's type. Thus, fixed points are unaffected. So when the designer wants to reduce manipulation costs, our main point on reducing allocation sensitivity from a fixed point is strengthened. But when the designer wants to induce higher actions, the result could reverse.

\subsection{Linear Policies} \label{sec:discussnonlinear}

That the designer uses linear allocation rules is generally restrictive. We nevertheless think it is interesting to focus on such policies for reasons beyond their analytical tractability.

First, linear policies are simple, canonical, and practical. As we have explained, they can be studied using linear regressions, which are widely used for estimating the relationship between agent characteristics and observable data. 

Second, linear policies are straightforward to interpret. Specifically, they are (up to a constant) fully ordered by the allocation's sensitivity to data. We can therefore discuss what it means it means for a policy to be ``flatter'', i.e., less sensitive to data, and compare the optimum to fixed points in this respect. Furthermore, it allows us to discuss how the designer optimally underutilizes data.

Third, and relatedly, linear policies are focal for comparison with linear fixed points. \citet{Gesche19} and \citet{FK19} have shown that fixing any linear strategy for the agent, the designer's best response is  linear if the agent's type distribution is bivariate elliptical \citep{GGM03}, subsuming bivariate normal; see also \citet{FV00} and \citet{BT06}. \cite{ball2020} extends these results to a multidimensional action space. Hence, under these joint distributions---and when agents optimally respond to linear allocation rules with linear strategies (see \autoref{sec:linassms})---linear fixed-point policies of the current paper are fixed points even when the designer and the agent choose arbitrary (possibly nonlinear) strategies $Y(x)$ and $X(\eta,\gamma)$.

\section{Conclusion}

To recap: how should a designer with market power choose an allocation rule to maximize allocation accuracy, given that the rule affects agents' manipulation of the very data used for allocation? Our main result is that the designer should commit to underutilizing information. The optimal allocation rule is flatter than a fixed-point, or ex-post optimal, rule.

Consider a large search engine like Google that aims to provide its users with an accurate list of organic search results. Our result suggests that Google's algorithm ought to be less responsive to webpage elements susceptible to search engine optimization (SEO) than if it were to take as given the extent of SEO. An analogous point holds for Amazon's use of consumer reviews.

Another application concerns standardized testing for college or high school admissions. Our result says that when students can study to the test, the accuracy of admissions can be improved by underutilizing the information contained in test results. However, when there are many decentralized colleges, we would expect each to have limited market power. Thus, each college would tend to use test information ex-post optimally, resulting in less-accurate admissions. One may extrapolate that the more competitive the college sector is, the lower the resulting accuracy. 

In various applications, the data that designers have access to is multidimensional. Our results would suggest that a designer should flatten the allocation rule more on dimensions that are more manipulable. Thus, observables that are difficult for the agent to manipulate become relatively more important for the designer's decision than those that are easy to manipulate. For instance, compared to the ex-post optimum, credit scores should put relatively more weight on the length of a consumer's credit history and less on the current credit utilization rate when the former is less manipulable. See \citet{ball2020} for some recent work in this direction.

A notable point is that information loss in our model stems from agents' heterogeneous responses to incentives. That is, it is not simply the existence of SEO and studying to the test that produces our implications for search engines and college admissions; rather, it is because some website managers engage in more SEO and some students have better access to test preparation. Having said that, the rationale for flattening allocations is not restricted to this source of information loss. For instance, even a model with a one-dimensional type (e.g., no heterogeneity on our gaming ability $\gamma$) may have information loss from ``pooling at the top'' in a bounded action space.  This could be relevant to test taking when there is a binding upper bound on the test score. \autoref{app:binary} establishes a version of our flattening result for a simple model in that vein.

As we have discussed, linear designer policies and agent responses are a benchmark for understanding how to improve (the use of) information from manipulable data. An interesting topic for future research is the generalization of ``flattening'' allocations and ``underutilizing'' information to nonlinear models.

\newpage

\appendix
\newpage
\counterwithin{lemma}{section}
\counterwithin{equation}{section}
\counterwithin{proposition}{section}
\counterwithin{claim}{section}
\counterwithin{remark}{section}

%{\Large \noindent \textbf{Appendices}}

\section{Appendix: Proofs}
\label{app:proofs}

\subsection{Proof of \autoref{lem:pos_beq}}
%As explained in \autoref{fn:anyvalue}, \autoref{eq:betahat} uniquely defines $\hat \beta(\beta)$ on $\beta\geq 0$ unless $\rho=-1$. If $\rho=-1$, then $\beq=\frac{\sigma_\eta}{m\sigma_\gamma}>0$ satisfies $\hat \beta(\beq) = \beq$ by the convention noted in \autoref{fn:anyvalue}. So assume $\rho>-1$. 
Recall from \autoref{sec:regression} that $\hat \beta(\beta)=\Cov(x,\eta)/\Var(x)$, with
\begin{align}
\Cov(x,\eta)&= \vareta + m \rho \sigma_\eta \sigma_\gamma \beta,\label{eq:cov}\\
\Var(x)&=\vareta + m^2  \vargamma \beta^2+ 2 m  \rho \sigma_\eta \sigma_\gamma \beta.\label{eq:var}
\end{align}
For $\beta\geq 0$, the fixed-point equation $\hat \beta(\beta)=\beta$ 
% For $\beta\geq 0$, \autoref{eq:betahat} 
can thus be rewritten as the cubic equation
\begin{align}
m^2 \vargamma \beta^3 + 2 m \rho \sigma_\eta \sigma_\gamma \beta^2 + (\vareta-m \rho \sigma_\eta \sigma_\gamma) \beta - \vareta = 0. \label{eq:eqcubic}
\end{align}
The left-hand side of \eqref{eq:eqcubic} is continuous, negative at $\beta=0$ and tends to $\infty$ as $\beta \to \infty$. There is a positive solution to \eqref{eq:eqcubic} by the intermediate value theorem. 

%For the second statement of the proposition, differentiate $\hat \beta(\cdot)$ from \autoref{eq:betahat} to obtain
%\[\hat \beta'(\beta) = - \frac{m \sigma_\eta \sigma_\gamma (2 \beta m \sigma_\eta \sigma_\gamma + \rho \sigma^2_\eta + \rho \beta^2 m^2 \sigma^2_\gamma)}{(\sigma^2_\eta + 2 \beta m \rho \sigma_\eta \sigma_\gamma + \beta^2 m^2 \sigma^2_\gamma)^2}.\]
%When $\rho\ge 0 $, this derivative is negative for all $\beta>0$. The result follows from the fact that $\hat \beta(0)=1$ and, when $\rho\geq 0$, $\hat \beta(1)<1$.

We now show that $\hat \beta(\beta)$ is strictly decreasing on $[0,\infty)$ when $\rho\geq 0$. Differentiating \autoref{eq:var} yields
\begin{equation}
\label{eq:var_deriv}
\frac{\mathrm d}{\mathrm d \beta}  \Var(x)= 2 m^2 \sigma^2_\gamma+2 m\rho \sigma_\eta \sigma_\gamma.	
\end{equation}
Dividing \autoref{eq:cov} by $\sigma_\eta\sqrt{\Var(x)}$ and differentiating yields
\begin{equation}
\label{eq:corr_deriv}
\frac{\mathrm d}{\mathrm d \beta}  \Corr(x,\eta)= - \frac{\beta m^2 (1-\rho^2)\sigma_\eta \sigma_\gamma^2}{\Var(x)^{3/2}}.	
\end{equation}
\autoref{eq:betahat} implies that $\hat \beta= \Corr(x,\eta) \sigma_\eta / \sqrt{\Var(x)}$. For $\rho\geq 0$ and $\beta\geq 0$, \autoref{eq:var_deriv} and \autoref{eq:corr_deriv} respectively show that $\Var(x)$ is strictly increasing and $\Corr(x,\eta)$ is strictly decreasing in $\beta$. The desired conclusion follows.

\subsection{Proof of \autoref{prop:main}} \label{appsec:proofmain}

From \autoref{sec:objective}, $(\beta^*,\beta^*_0)$ solves
\[ 
\min_{(\beta,\beta_0) \in \Reals^2} \ \mathbb{E}[(m \beta^2  \gamma + \beta_0-(1-\beta)\eta)^2]. \]
The first-order condition with respect to $\beta_0$ implies
\begin{align*}
\beta_0^* &%= \mathbb{E} [ (1-\beta) \eta - m \beta^2 \gamma] 
= (1-\beta) \mu_ \eta -m \beta^2 \mu_ \gamma.% \label{eq:beta0val}
\end{align*}
Substituting $\beta_0^*$ into the objective, the designer chooses $\beta$ to minimize
\begin{align}
&\mathbb{E} [(m \beta^2 ( \gamma - \mu_ \gamma) -(1-\beta) (\eta-\mu_ \eta))^2] \notag
\\
%= &(1-\beta)^2 \vareta +m^2 \beta^4 \vargamma - 2 (1-\beta) m \beta^2 \cov(\eta,\gamma) \notag
%\\
=& (1-\beta)^2 \vareta +m^2 \beta^4 \vargamma - 2 (1-\beta) m \beta^2 \rho \sigma_\eta \sigma_\gamma \notag \\
= & \sigma^2_\eta \left[\left((1-\beta)-k\beta^2\right)^2+2(1-\rho)\beta^2(1-\beta)k\right],\notag
%\label{eq:Wbeta}
\end{align}
where
$$ k \equiv m \sigma_\gamma/\sigma_\eta>0.$$
Equivalently, for $k>0$ and $\rho \in (-1,1)$, $\beta^*$ minimizes
\begin{equation}
\loss(\beta,k,\rho) \equiv \left(k\beta^2+\beta -1\right)^2+2(1-\rho)\beta^2(1-\beta)k.\label{eq:L}	
\end{equation}
Differentiating,
\begin{align}
{\loss}_{\beta}(\beta,k,\rho) = -2(1-\beta)+4 k^2 \beta^3+2 \rho k \beta(3\beta-2).\label{eq:Ldbeta}
%\\ = 4 m^2 \vargamma \beta^3 + 6 m \rho \sigma_\eta \sigma_\gamma \beta^2 - (4 m \rho \sigma_\eta \sigma_\gamma - 2 \vareta)\beta - 2 \vareta.
%\label{eq:Wprimebeta}
\end{align}
Note that $L_{\beta}(0,k,\rho)<0$, i.e., there is a first-order benefit from putting some positive weight on the agent's action.%, i.e., increasing $\beta$ from zero. %Furthermore, $L_{\beta}(\beta,k,\rho)\to \infty$ as $\beta \to \infty$. %We illustrate the function $L_{\beta}$ in \autoref{fig:Lbeta}.

%\begin{figure}
%\begin{center}
%\includegraphics[width=4in]{Lbeta.pdf}
%\end{center}
%\caption{\label{fig:Lbeta} The function $L_\beta$ from \autoref{eq:Ldbeta}. Parameters: $k=0.15$ and $\rho=-0.99$.}
%\end{figure}

The last statement of \autoref{prop:main} follows from the second because, from \autoref{eq:betahat}, $\hat \beta(\cdot)$ is continuous, $\hat \beta(0)>0$, and $\hat \beta(\beq)=\beq$ for any $\beq$. \autoref{prop:main} is thus implied by \autoref{lem:lossderivbeq} and the following result. We abuse notation hereafter and drop the arguments $k$ and $\rho$ from $L(\cdot)$ when those are held fixed. So, for example, $L(\beta)$ means that both $k$ and $\rho$ are fixed.

\begin{lemma}
\label{lem:bstarsoln}
There exists $\bstar \in (0,2)$  such that:\reducespace
\begin{enumerate}
\item The loss function $L(\beta)$ from \eqref{eq:L} is uniquely minimized over $\beta \in \mathbb{R}$ at $\beta^*$.
\label{lem:betapos}
\item \label{bstarfirstzero} $\beta^* = \min_{\beta\ge0}\{ \beta : L'(\beta) \ge 0\}%=\min_{\beta>0}\{\beta | L'(\beta) =0\}
$.
\item $L''(\bstar)>0$.
\end{enumerate}
\end{lemma}

%\begin{remark}
%\label{rem:bstarunique}
%In the proof below, we also allow for $\rho=-1$ and $\rho=1$. The proof establishes that even in these boundary cases: (i) there are at most two minimizers of $L(\beta)$; (ii) when $\rho=1$, there is another minimizer besides $\bstar$, which is negative; (iii) when $\rho=-1$, there is another minimizer if and only if $k<1/4$, in which case it is larger than 2.
%\end{remark}

\begin{proof}%[Proof of \autoref{lem:bstarsoln}]
The proof has five steps below. Steps 1--3 are building blocks to Step 4, which establishes that all minimizers of $L(\beta)$ are in $(0,2)$. Step 5 then establishes there is in fact a unique minimizer, and it has the requisite properties. It is useful in this proof to extend the domain of the function $L$ defined in \eqref{eq:L} to include $\rho=-1$ and $\rho=1$.

\underline{Step 1:} We first establish two useful properties of $L(\beta,\rho=1)$. Simplifying \eqref{eq:L},
\begin{equation*}
	%\label{eq:Wbeta1}
L(\beta,\rho=1) = \left(k \beta^2+\beta -1\right)^2
\end{equation*}
is the square of a quadratic. The quadratic $k \beta^2+\beta -1$ is strictly convex in $\beta$, minimized at 
\begin{equation}
	\beta=\beta^m \equiv -1/(2k)<0,\label{eq:beta^m}
\end{equation}
and, because it has one negative and one positive root, it is negative and strictly increasing on $[\beta^m,0]$. 
%\begin{equation*}
%\beta=\bar \beta \equiv \frac{-1+ \sqrt{1 + 4 k}}{2 k} \in (0,2).%\label{eq:beta*rho1}
%\end{equation*}
It follows that $L(\cdot,\rho=1)$ is strictly decreasing on $[\beta^m,0]$ and symmetric around $\beta^m$ (i.e., for any $x$, $L(\beta^m+x,\rho=1)=L(\beta^m-x,\rho=1)$).

%and the function has two roots, one of which is negative and the other is 
%\begin{equation*}
%\beta=\bar \beta \equiv \frac{-1+ \sqrt{1 + 4 k}}{2 k} \in (0,2).%\label{eq:beta*rho1}
%\end{equation*}
%Moreover, the quadratic expression $k\beta^2+\beta-1$ is negative and strictly increasing on $[0,\beta^*)$, hence $L'(\beta)<0$ on this region.
%So $L(\cdot,\rho=1)$ is minimized at $\bar \beta$, and there is no other nonnegative minimizer. Moreover, $L(\cdot,\rho=1)$ is strictly quasiconvex on $(-\infty,\beta^m]$, strictly decreasing on $[\beta^m,\bar \beta]$, and strictly increasing on $[\bar \beta,\infty)$. Still further, $L(\cdot,\rho=1)$ is symmetric around $\beta^m$: for any $x\in \Reals$, $L(\beta^m+x,\rho=1)=L(\beta^m-x,\rho=1)$.

\underline{Step 2:} We claim that for any $\beta<0$ and $\rho<1$, there is $\tilde \beta\geq 0$ such that $L(\tilde \beta)<L(\beta)$. Since $L'(0)<0$, %---and hence $L(\beta)$ is not minimized at $\beta=0$---
it follows that for $\rho<1$, $\argmin L(\beta,\rho) \subset \Reals_{++}$.

To prove the claim, we first establish that for any $x>0$ and $\beta=\beta^m-x$ (where $\beta^m$ is defined in \eqref{eq:beta^m}), the symmetric point $\beta^m+x$ has a lower loss when $\rho<1$; note that $\beta^m+x$ may also be negative. 
The argument is as follows:
\begin{align*}
L(\beta^m-x,\rho)-L(\beta^m+x,\rho)&=L(\beta^m-x,\rho=1)+2(1-\rho)(\beta^m-x)^2(1-\beta^m+x)k\\ & \quad - \left[L(\beta^m+x,\rho=1)+2(1-\rho)(\beta^m+x)^2(1-\beta^m-x)k\right]\\
&=2(1-\rho)k\left[(\beta^m-x)^2(1-\beta^m+x)-(\beta^m+x)^2(1-\beta^m-x) \right]\\
&= 4(1-\rho)k x \left(\beta^m (3 \beta^m-2)+x^2\right)\\
&\geq 0,
\end{align*}
where the first equality is from the definition in \eqref{eq:L}, the second is because Step 1 established that $L(\beta^m+x,\rho=1)=L(\beta^m-x,\rho=1)$, the third equality is from algebraic simplification, and the inequality is because $\beta^m<0$, $x>0$, and $\rho<1$.

It now suffices to establish $L(0,\rho)<L(\beta,\rho)$ for all $\beta \in [\beta^m,0)$. Differentiating \eqref{eq:Ldbeta} yields $L_{\beta \rho}(\beta,\rho) = 2 k \beta (3\beta-2) > 0$ when $\beta<0$.
Hence for $\beta\in [\beta^m,0)$, $L(0,\rho)-L(\beta,\rho)\leq L(0,\rho=1)-L(\beta,\rho=1)<0$, where the strict inequality is from Step 1.%  The claim follows.

%Accordingly, we may hereafter restrict attention to $\beta>0$.

\underline{Step 3:} $\argmin_\beta L(\beta,\rho=-1) \cap (0,2] \neq \emptyset$. 
%If $k\neq 1/4$, then $\argmin_\beta L(\beta,\rho=-1) \cap (0,2) \neq \emptyset$. If $k=1/4$, then $\argmin_\beta L(\beta,\rho=-1)=\{2\}$.

To prove this, simplify \eqref{eq:L} to get
$$L(\beta,\rho=-1)=\left(k \beta^2-\beta+1\right)^2.$$
The quadratic $k\beta^2-\beta+1$ is strictly convex in $\beta$ and minimized at $\beta=1/(2k)$; moreover, if $k\geq 1/4$ then that quadratic is nonnegative, and otherwise it is equal to zero at $\beta=\frac{1\pm \sqrt{1-4k}}{2k}$. 
It follows that if $k\geq 1/4$, $\argmin L(\beta,\rho=-1)=\{1/(2k)\}\subset (0,2]$%, and hence the unique minimizer is in $(0,2]$% (in the interior unless $k=1/4$)
.
If $k\in (0,1/4)$, $\min \argmin L(\beta,\rho=-1)= \frac{1-\sqrt{1-4k}}{2k}  \in (0,2)$.%\footnote{The value $\frac{1-\sqrt{1-4k}}{2k}$ is clearly positive. To see that $\frac{1-\sqrt{1-4k}}{2k}<2$, observe that this inequality is equivalent to $1-4k< \sqrt{1-4k}$, which holds because $1-4k \in (0,1)$.}

\underline{Step 4:} For $\rho \in (-1,1)$, $\argmin_{\beta} L(\beta,\rho) \subset (0,2)$.

To prove this, note that $L_{\beta \rho}(\beta,\rho)=2k\beta(3\beta-2)>0$ when $\beta>2/3$. Monotone comparative statics (see \autoref{MCS} in the Supplementary Appendix) imply that on the domain $(2/3,\infty)$ every minimizer of $L(\cdot,\rho)$ when $\rho>-1$ is smaller than every minimizer of $L(\cdot,\rho=-1)$. Step 3 then implies that all minimizers when $\rho>-1$ are less than $2$; Step 2 established that when $\rho<1$, all minimizers are larger than $0$.

\underline{Step 5:} Finally, we claim that for $\rho\in (-1,1)$, $L'(\beta)$ has only one root in $(0,2)$; moreover, $L''(\beta)>0$ at that root.
%$\min_{\beta \in (0,2)}\{\beta|L'(\beta)=0\}$ satisfies $L''(\beta)\geq 0$, and any other root of $L'(\beta)$ in $(0,2)$ has $L''(\beta)<0$ (and thus cannot be a local minimizer).
%is a unique solution to $L{''}(\beta)\geq 0$ and $L'(\beta)=0$. This implies that there is a unique solution on $(0,2)$ to the FOC $L'(\beta)=0$; denote this global minimizer $\beta^*$. %\footnote{For either $\rho>-1$ or $k\neq 1/4$, Steps 3 and 4 established that the minimizer within $[0,2]$ is in the open set $(0,2)$, so the FOC is necessary. For $\rho=-1$ and $k=1/4$, Step 3 established that the minimizer is $2$, which does in fact satisfy the FOC.} 
The lemma follows because $L'(\beta)$ is continuous and \mbox{$L'(0)<0$}.

To prove the claim, first observe from \autoref{eq:Ldbeta} that %the cubic function
%$$L'(\beta)=4 k^2 \beta^3+6 \rho k \beta^2+(2-4 \rho k) \beta-2$$
$L'(\beta)$ is a cubic function that is initially strictly concave and then strictly convex, with inflection point \mbox{$\beta=-\rho/(2k)$}. For the rest of the proof, view $L'$ or $L''$ as a function of $\beta$ only.
\reducespace
\begin{enumerate}	
\item If $\rho\geq 0$, then the inflection point is negative, and thus $L'$ is strictly convex on $\beta>0$. Since $L'(0)<0$, $L'$ has only one positive root, and $L''>0$ at that root.
\item Consider $\rho\in (-1,0)$. 
$L''$ is minimized at the inflection point of $L'$. Differentiating \autoref{eq:Ldbeta} and evaluating at the inflection point,
$$L''\left(\frac{-\rho}{2k}\right)=2+12k^2\left(\frac{-\rho}{2k}\right)^2+4\rho k\left(3\left(\frac{-\rho}{2k}\right)-1\right)=2-3\rho^2-4k\rho.$$ 
If this expression is positive, then $L''(\beta)>0$ for all $\beta$, i.e., $L'$ is strictly increasing and hence has a unique root.
%: $L'$ is concave and increasing on $\beta$ below the inflection point, and is convex and increasing on $\beta$ to the right of the inflection point. Since $L'$ is strictly increasing it has a unique root.

So suppose instead %that the derivative of $L'$ is not positive at the inflection point, i.e., 
$2-3\rho^2-4k\rho\leq 0$. Equivalently, since $\rho<0$, suppose 
$$k \leq \frac{2 - 3 \rho^2}{4 \rho}.$$
The right-hand side of this inequality is less than $-\rho/4$ because $\rho\in (-1,0)$,
and hence $k<-\rho/4$. Consequently, the inflection point,  $\beta=-\rho/(2k)$, is larger than $2$, and therefore $L'(\beta)$ is concave over $\beta \in (0,2)$. Moreover,  recall that $L'(0)<0$, and also observe that $L'(2)=32 k^2 + 16 k \rho +2>0$ because $k<-\rho/4$ and $\rho \in (-1,0)$. %\footnote{To see that $L'(2)>0$, observe that the quadratic expression $32 k^2 + 16 k \rho +2$ is minimized over choice of $k$ at $k=-\rho/4$, at which point its value is $-2\rho^2+2$. Since $k<-\rho/4$,  we have $L'(2)>-2 \rho^2+2$, and this right-hand side is larger than 0 because $\rho \in (-1,0)$.} 
It follows that $L'$ has only one root on $(0,2)$, and $L''>0$ at that root.\qedhere
\end{enumerate}
\end{proof}

\begin{claim}
\label{claim:bstarnaive}
It holds that $\bstar <1$ if $2k>-\rho$, $\bstar>1$ if $2k<-\rho$, and $\bstar=1$ if $2k=-\rho$.
\end{claim}
\begin{proof}
\autoref{eq:Ldbeta} yields $L_\beta(1,k,\rho)=4k^2+2k\rho$. So $\sign[L_\beta(1,k,\rho)]=\sign[2k+\rho]$. As $L_\beta(0,k,\rho)<0$, the result follows from the fact that $L_\beta(\cdot,k,\rho)$ is continuous and has only one root in $(0,2)$, which is $\bstar$ (Step 5 in the proof of \autoref{lem:bstarsoln}).
\end{proof}

\subsection{Proof of \autoref{lem:lossderivbeq}}
\label{proof:lossderivbeq}
As explained in \autoref{sec:regression},
\begin{align*}
\mathbb{E}[(\hat \eta_{\beta}(x)-\eta)^2] &= \sigma^2_\eta \left(1-R^2_{\eta x}\right)=\sigma^2_\eta\left(1-\Corr(x,\eta)^2\right).
\end{align*}
We also have
\begin{align*}
\mathbb{E}[(Y(x) - \hat \eta_{\beta}(x))^2] &= \E[(\beta x+\beta_0 -\hat \beta(\beta)x-\hat \beta_0(\beta))^2] \qquad \text{\small from definitions}\\
&=\E\left[\left((\beta-\hat \beta(\beta))(x-\E[x])\right)^2\right]\\
&=(\beta-\hat \beta(\beta))^2\var(x),
\end{align*}
where the second line is because %the choice of $\beta_0$ and $\hat \beta(\beta)$ are such that 
$\beta \E[x]+\beta_0=\mu_\eta=\hat \beta(\beta)\E[x]+\hat \beta_0(\beta)$ (the second equality here is standard; for the first, see the beginning of the proof of \autoref{prop:main}) and hence $\beta_0-\hat \beta_0(\beta)=(\hat \beta(\beta)-\beta)\E[x]$.

Substituting these formulae into \autoref{eq:twopartslin} yields
%$$\welfloss(\beta) = \underbrace{\sigma^2_\eta -(\hat \beta(\beta))^2 \var(x)}_{\text{Info loss}}+\underbrace{(\beta-\hat \beta(\beta))^2\var(x)}_{\text{Misallocation loss}}.$$
$$\welfloss(\beta) = \underbrace{\sigma^2_{\eta}\left(1-\Corr(x,\eta)^2\right)}_{\text{Info loss from linearly estimating $\eta$ using $x$}}+\underbrace{(\beta-\hat \beta(\beta))^2\var(x)}_{\text{Misallocation loss given linear estimation}}.$$
Differentiating,
\begin{align*}
\welfloss'(\beta) &= \overbrace{- 2 \sigma^2_\eta \Corr(x,\eta) \frac{\mathrm d }{\mathrm d \beta} \Corr(x,\eta)}^{\text{Marginal change in info loss}} \\ &\qquad + \underbrace{\left( -2 (\beta - \hat \beta(\beta)) \hat \beta'(\beta) \var(x) +   (\beta - \hat \beta(\beta))^2 \frac{\mathrm d}{\mathrm d \beta} \var(x)\right)}_{\text{Marginal change in misallocation loss}}.
\end{align*}
Let us evaluate this expression at $\beta=\beq$. 
Since $\beq=\hat \beta(\beq)$, the marginal change in misallocation loss is evidently zero. % (intuitively because the misallocation loss is minimized at $\beta=\beq$). 
Thus,
\begin{align*}
%\label{eq:welfloss'1}
\sign\welfloss'(\beq) &= \left.- \left(\sign \Corr(x,\eta)\right) \left( \sign \frac{\mathrm d }{\mathrm d \beta} \Corr(x,\eta)\right)\right|_{\beta=\beq}\\
&= -(\sign \beq)(-\sign \beq)\\
&>0,
\end{align*}
where the second equality is because $\beq\neq 0$, $\sign \left.\Corr(x,\eta)\right|_{\beta=\beq}=\sign \hat \beta(\beq)=\sign \beq$ (see \autoref{eq:betahat}), and 
$\left.\sign \frac{\mathrm d }{\mathrm d \beta} \Corr(x,\eta)\right|_{\beta=\beq} = - \sign \beq$ (see \autoref{eq:corr_deriv}).

\subsection{Proof of \autoref{prop:compstats}}

The proof is via the following claims. Applying \autoref{lem:bstarsoln}, we without loss restrict attention to $\beta\in (0,2)$ in all the claims.

\begin{claim}
$\beta^*$ is continuously differentiable in $\rho$ and $k$.	
\end{claim}
\begin{proof}
\autoref{lem:bstarsoln} established that $\sign[L''(\beta^*)]>0$. Thus, the implicit function theorem guarantees the existence of 
$\frac{\mathrm d \beta^*}{\mathrm d k} = -\frac{L_{\beta k}}{L_{\beta \beta}}$ and $\frac{\mathrm d \beta^*}{\mathrm d \rho} = -\frac{L_{\beta \rho}}{L_{\beta \beta}}$.
\end{proof}

\begin{claim}
If $k>3/4$ then $\beta^*<2/3$ and is strictly increasing in $\rho$. 
If $k<3/4$ then $\beta^*>2/3$ and is strictly decreasing in $\rho$. 
If $k=3/4$ then $\beta^*=2/3$ independent of $\rho$.
\end{claim}
\begin{proof}
From \autoref{eq:Ldbeta} compute the cross partial
$$L_{\beta \rho}(\cdot) = 2 k \beta (3\beta-2).$$
Hence $L_{\beta \rho}<0$ when $\beta < 2/3$, while $L_{\beta \rho}>0$ when $\beta>2/3$. Moreover, it follows from \autoref{eq:Ldbeta} that when $\beta=2/3$, $\sign[L_\beta]=\sign[k-3/4]$ independent of $\rho$. %Note that $L_\beta=0$ when $\beta=\beta^*$ (\autoref{lem:bstarsoln}).
%\reducespace
\begin{enumerate}
	\item Consider $k=3/4$. Routine algebra verifies that $L_{\beta}$ is strictly increasing in $\beta$, and hence $L_\beta=0 \implies \beta=2/3$, i.e., $\beta^*=2/3$ independent of $\rho$.
	\item Consider $k>3/4$. Since $L_\beta>0$ when $\beta=2/3$, it follows that $\beta^*<2/3$. (Recall $L_\beta<0$ when $\beta=0$, and \autoref{lem:bstarsoln} implies that $\beta^*=\min\{\beta>0:L_\beta=0\}$.) Since $L_{\beta \rho}<0$ on the domain $\beta<2/3$, monotone comparative statics (see \autoref{MCS} in the Supplementary Appendix) imply $\beta^*$ is strictly increasing in $\rho$.
	\item Consider $k<3/4$. For $\rho=0$, we have $L_{\beta k} = 8k\beta^3>0$ and hence $\beta^*>2/3$ using $\beta^*=2/3$ when $k=3/4$ and monotone comparative statics. It follows that $\beta^*>2/3$ for all $\rho$ because $\beta^*$ is continuous in $\rho$ and $L_\beta<0$ when $\beta=2/3$ whereas $L_\beta=0$ when $\beta=\beta^*$.
%$$\frac{\mathrm d \beta^*}{\mathrm d k}=-\frac{L_{\beta k}}{L_{\beta \beta}}=-\frac{8k\beta^3+2\rho \beta(3\beta-2)}{2+12k^2\beta^2+4\rho k(3\beta-1)}<0,$$
%where the inequality holds on the domain $\beta \geq 2/3$.
%		Applying the implicit function theorem again, $\frac{\mathrm d \beta^*}{\mathrm d \rho}=-\frac{L_{\beta \rho}}{L_{\beta \beta}}<0$, because $L_{\beta \beta}= 2+12 k^2\beta^2+4\rho k(3\beta-1)>0$ when $\rho \geq 0$ and $\beta>2/3$, and $L_{\beta \rho}>0$ when $\beta>2/3$ was noted earlier.
Since $L_{\beta \rho}>0$ on the domain $\beta>2/3$, monotone comparative statics imply $\beta^*$ is strictly decreasing in $\rho$. \qedhere

\end{enumerate}
\end{proof}

\begin{claim}
As $k\to \infty$, $\bstar \to 0$; as $k\to 0$, $\bstar \to 1$. If $\rho\geq 0$ then $\bstar$ is strictly decreasing in $k$. If $\rho<0$ then $\bstar$ is strictly quasiconcave in $k$, attaining a maximum at some point.
\end{claim}

\begin{proof}
The first statement about limits is evident from inspecting \autoref{eq:Ldbeta}, since $L_\beta(\bstar,\cdot)=0$. For the comparative statics, first compute the cross partials
\begin{align}
L_{\beta k}(\beta,k) & = 2\beta[4k\beta^2+\rho(3\beta-2)],\label{eq:Ldbetak}\\
L_{\beta k k}(\beta,k) & = 8 \beta^3.\label{eq:Ldbetakk}
\end{align}
Below we write $\bstar(k)$ to make the dependence on $k$ explicit, and $\bstarprime(k)$ for the derivative. Furthermore, let $\hstar(k)\equiv L_{\beta k}(\bstar(k),k)$, with
\begin{equation}
\hstarprime(k)=L_{\beta k k}(\bstar(k),k)+L_{\beta k \beta}(\bstar(k),k)\bstarprime(k).\label{eq:Lstardk}
\end{equation}
Note that $\sign[\hstar(k)]=-\sign[\bstarprime(k)]$, since 
\begin{equation}
\bstarprime(k)=-\frac{\hstar(k)}{L_{\beta \beta}(\bstar(k),k)},	\label{eq:bstarprime}
\end{equation} 
and \autoref{lem:bstarsoln} established that $L_{\beta \beta}(\bstar(k),k)>0$. 

\underline{Fact}: $\bstar(\cdot)$ is strictly quasiconcave: if $\bstarprime(\hat k)=0$, then $\bstarprimeprime(\hat k)<0$. To prove this, assume $\bstarprime(\hat k)=0$, or equivalently, $\hstar(\hat k)=0$. The derivative of the right-hand side of \eqref{eq:bstarprime} with respect to $k$, evaluated at $\hat k$, has the same sign as $-\hstarprime(\hat k)$ (using $h^*(\hat k)=0$ and $L_{\beta \beta}(\bstar(\hat k),\hat k)>0$). Hence, $\bstarprimeprime(\hat k)<0\iff \hstarprime(\hat k)>0$. The latter inequality follows from \eqref{eq:Lstardk} and \eqref{eq:Ldbetakk}, using $\bstarprime(\hat k)=0$ and $\bstar(\cdot)>0$ (by \autoref{lem:bstarsoln}). $\parallel$

This Fact implies the desired comparative statics as follows:
\begin{enumerate}
	\item \label{rhopos} Assume $\rho\geq 0$. As $k\to 0$, $\bstar(k)\to 1$ and  \eqref{eq:Ldbetak} implies $L_{\beta k}(\bstar(k),k)>0$, hence $\bstarprime(k)<0$. Since $\bstar(\cdot)$ is strictly quasiconcave, $\bstarprime(k)<0$ for all $k$.
	\item Assume $\rho<0$. As $k\to 0$, $\bstar(k)\to 1$ and \eqref{eq:Ldbetak} implies $L_{\beta k}(\bstar(k),k)<0$, hence $\bstarprime(k)>0$. As $k\to \infty$, $\bstar(k)\to 0$. Hence, there is some $\hat k$ at which $\bstarprime(\hat k)=0$. Since $\bstar(\cdot)$ is strictly quasiconcave, $\hat k$ is unique, with $\bstarprime(k)>0$ for $k<\hat k$ and $\bstarprime(k)<0$ for $k>\hat k$. \qedhere
\end{enumerate}

\end{proof}

\begin{claim}
Assume $\rho=0$. There is a unique $\beq$, which is positive. Both $\beq$ and $\bstar/\beq$ are strictly decreasing in $k$. Moreover, $\bstar/\beq\to 1$ as $k\to \infty$ and $\bstar/\beq\to \sqrt[3]{1/2}$ as $k\to 0$.
\end{claim}
\begin{proof}
Assume $\rho=0$. \autoref{eq:eqcubic} simplifies to
\begin{comment}
$\beq \in (0,1)$ is the unique positive $\beta$ that solves  
\begin{align*}
m^2 \sigma_\gamma^2 \beta^3 + 2 m  \rho \sigma_\gamma \sigma_\eta \beta^2 + (\sigma_\eta^2 - m \rho \sigma_\eta \sigma_\gamma) \beta -\sigma_\eta^2=0.
%\label{e:L-root-eta}
\end{align*}
By implicit differentiation, using formulae from \citet{FK19},
\begin{align*}
\frac{\mathrm d \beq}{\mathrm d m} & =-\frac{\beq \sigma_\gamma[2(\beq)^2 m \sigma_\gamma+(2\beq-1)\rho \sigma_\eta]}{\sigma^2_\eta+3 (\beq)^2 m^2\sigma_\gamma^2+(4\beq-1)\rho m \sigma_\eta\sigma_\gamma},\\[10pt]
\frac{\mathrm d \beq}{\mathrm d \sigma_\gamma} & =-\frac{\beq m[2(\beq)^2 m \sigma_\gamma+(2\beq-1)\rho \sigma_\eta]}{\sigma^2_\eta+3 (\beq)^2 m^2\sigma_\gamma^2+(4\beq-1)\rho m \sigma_\eta\sigma_\gamma},\\
\frac{\mathrm d \beq}{\mathrm d \sigma_\eta} & =-\frac{2m\rho \sigma_\gamma (\beq)^2+(2 \sigma_\eta-m \rho \sigma_\gamma)\beq-2\sigma_\eta}{\sigma^2_\eta+3 (\beq)^2 m^2\sigma_\gamma^2+(4\beq -1)\rho m \sigma_\eta\sigma_\gamma}.
\end{align*}
\citet{FK19} establish that these denominators are positive. Plainly, when $\rho=0$, the first two numerators are positive, while the last numerator is negative. Hence when $\rho=0$, $\beq$ is strictly decreasing in $k$.
\end{comment}
\begin{equation}
\label{eq:beq-rho0}
k^2 (\beq)^3 + \beq -1=0,	
\end{equation}
which has a unique solution, with $\beq\in(0,1)$ strictly decreasing in $k$ with range $(0,1)$.

The first-order condition for $\bstar$ simplifies to
\begin{equation}
\label{eq:bstar-rho0}
2 k^2 (\bstar)^3 + \bstar - 1=0,
\end{equation}
which has a unique solution, also in $(0,1)$ and strictly decreasing in $k$ with range $(0,1)$.

Hence, $\bstar/\beq\to 1$ as $k\to 0$. Moroever, \autoref{eq:beq-rho0} and \autoref{eq:bstar-rho0} imply that as $k\to \infty$, $k^2 (\beq)^3\to 1$ and $2 k^2 (\bstar)^3 \to 1$, and hence $(\bstar/\beq)\to \sqrt[3]{1/2}$.

It remains to prove that $\bstar/\beq$ is strictly decreasing in $k$. Applying the implicit function theorem to \autoref{eq:beq-rho0} and \autoref{eq:bstar-rho0} (which is indeed valid) and doing some algebra,
\begin{align*}
\frac{\mathrm d \bstar}{\mathrm d k}&=-\frac{4 k (\bstar)^3}{6k^2(\bstar)^2+1},\\[5pt]	
\frac{\mathrm d \beq}{\mathrm d k}&=-\frac{2 k (\beq)^3}{3k^2(\beq)^2+1}.
\end{align*}
The ratio $\bstar/\beq$ is strictly decreasing in $k$ if and only if $\beq \frac{\mathrm d \bstar}{\mathrm d k}-\bstar \frac{\mathrm d \beq}{\mathrm d k}<0$. Substituting in the formulae above, this inequality is equivalent to
\begin{align*}
&\frac{2 k (\beq)^3\bstar}{3k^2(\beq)^2+1}<\frac{4 k (\bstar)^3 \beq}{6k^2(\bstar)^2+1}\\[5pt]
\iff & \left(6k^2(\bstar)^2+1\right)(\beq)^2 < \left(3k^2(\beq)^2+1\right)2(\bstar)^2\\[5pt]
\iff & \beq < \bstar \sqrt{2}.
\end{align*}
Plainly, the last inequality holds as $k\to 0$ because both $\beq\to 1$ and $\bstar\to 1$ as $k \to 0$. By continuity, we are done if there is no $k$ at which $\beq=\bstar \sqrt{2}$. Indeed there is not because then \autoref{eq:beq-rho0} would become equivalent to
$$2 k^2 (\bstar)^3+\bstar-1/\sqrt{2}=0,$$
contradicting \autoref{eq:bstar-rho0}.
\end{proof}

\subsection{Proof of \autoref{prop:compstatswelfare}}

Recall $L(\beta,k,\rho)$ defined in \eqref{eq:L} and that $k\equiv m\sigma_\gamma/\sigma_\eta>0$. As explained before \eqref{eq:L}, the welfare loss at $\beta$ is $\sigma^2_\eta L(\beta,k,\rho)$. Thus, the welfare loss' comparative statics in $\sigma_\gamma$, $m$, and $\rho$ are given by those of $L(\beta^*,k,\rho)$ in $k$ and $\rho$. Although $\bstar$ depends on $k$ and $\rho$, the envelope theorem implies that $$\frac{\mathrm dL(\bstar,k,\rho)}{\mathrm d k}=L_k(\bstar,k,\rho) \text{ \ and \ } \frac{\mathrm dL(\bstar,k,\rho)}{\mathrm d \rho}=L_\rho(\bstar,k,\rho).$$

Hence, the proposition's comparative statics in $\sigma_\eta$, $m$, and $\rho$ follow from:

\begin{lemma}
\label{lem:dLdk}
$L_{k}(\bstar,k,\rho)>0$. If $\rho\geq 0$, then $L_{\rho}(\bstar,k,\rho)<0$.
\end{lemma}

\begin{proof}
The second statement follows from the fact that $\bstar\in (0,1)$ when $\rho\geq 0$ (\autoref{rem:bstar<1}), and the computation $$L_{\rho}(\bstar,k,\rho)=-2(\bstar)^2(1-\bstar)k.$$
So we are left to prove $L_{k}(\bstar,k,\rho)>0$. Compute $L_{k}(\bstar,k,\rho)=2(\bstar)^2\left(k(\bstar)^2-\rho(1-\bstar)\right).$ Letting
%implies $\sign L_k(\bstar,k,\rho)=\sign\left( k(\bstar)^2-\rho(1-\bstar)\right)$.
$$%\label{eq:xi}
\xi(\beta,k,\rho)\equiv k \beta ^2-\rho(1-\beta),	
$$
the fact that $\bstar>0$ implies that $L_k(\bstar,k,\rho)>0$ is equivalent to
\begin{equation}
\xi(\bstar,k,\rho)>0.\label{eq:dLdk_step}
\end{equation}
%Consider two exhaustive cases.

Since $\xi(\beta,k,\rho)$ is a convex quadratic in $\beta$, \eqref{eq:dLdk_step} holds if $\xi(\cdot,k,\rho)$ has no real roots. So consider the case that the real roots $\frac{-\rho\pm \sqrt{\rho^2+4 k \rho}}{2k}$ exist, i.e., 
\begin{equation}
\rho(4k+\rho)\geq 0.\label{eq:dLdk_realroots}	
\end{equation}
 Then 
 \begin{equation}
 \label{eq:dLdk_bstarthresh}
 \eqref{eq:dLdk_step}  \iff \bstar \notin \left[\frac{-\rho - \sqrt{\rho^2+4 k \rho}}{2k},\frac{-\rho+\sqrt{\rho^2+4 k \rho}}{2k}\right].
 \end{equation}
  Evaluating \autoref{eq:Ldbeta} at $\beta=\frac{-\rho\pm \sqrt{\rho^2+4 k \rho}}{2k}$ and simplifying yields
\begin{align}
L_\beta\left(\frac{-\rho - \sqrt{\rho^2+4 k \rho}}{2k},k,\rho\right)&=-\frac{\left(1-\rho^2\right) \left(2 k+\rho+\sqrt{\rho (4 k+\rho)}\right)}{k}\label{eq:dLdk_step2},	\\
L_\beta\left(\frac{-\rho + \sqrt{\rho^2+4 k \rho}}{2k},k,\rho\right)&=-\frac{\left(1-\rho^2\right) \left(2 k+\rho-\sqrt{\rho (4 k+\rho)}\right)}{k}.\label{eq:dLdk_step3}
\end{align}
%Since $\bstar$ is the first nonnegative point at which $L_\beta(\cdot,k,\rho)\geq 0$ (by \autoref{lem:bstarsoln} part \ref{bstarfirstzero} and $L_\beta(0,k,\rho)<0$), \eqref{eq:dLdk_step} holds if and only if either \eqref{eq:dLdk_step2} is positive or \eqref{eq:dLdk_step3} is negative. One of these is the case because
Note that
\begin{equation}
(2k+\rho)^2=4k^2+4k\rho +\rho^2>\rho (4 k+\rho).\label{eq:dLdk_step4}
\end{equation}
%where the strict inequality is because $k>0$ and the weak inequality is \eqref{eq:dLdk_realroots}.

Consider $\rho\geq 0$. Then $2k+\rho>0$, which combines with \eqref{eq:dLdk_step4} and $\rho^2<1$ to imply \eqref{eq:dLdk_step3} is negative. Since $L_\beta(\cdot,k,\rho)$ has only one positive root when $\rho\geq 0$ (see Step 5 in the proof of \autoref{lem:bstarsoln}), and this root is $\bstar$, it follows that $\bstar>\frac{-\rho+\sqrt{\rho^2+4 k \rho}}{2k}$. So \eqref{eq:dLdk_bstarthresh} implies that \eqref{eq:dLdk_step} holds.

Consider $\rho<0$. Then \eqref{eq:dLdk_realroots} implies $2k+\rho<4k+\rho\leq 0$, which combines with \eqref{eq:dLdk_step4} and $\rho<1$ to imply that \eqref{eq:dLdk_step2} is positive. Hence, since $\bstar$ is the first nonnegative point at which $L_\beta(\cdot,k,\rho)\geq 0$ (\autoref{lem:bstarsoln} part \ref{bstarfirstzero}), $\bstar<\frac{-\rho-\sqrt{\rho^2+4 k \rho}}{2k}$. So \eqref{eq:dLdk_bstarthresh} implies that \eqref{eq:dLdk_step} holds.
\end{proof}

It remains to show the comparative statics in $\sigma_\eta$. As explained before \eqref{eq:L}, the welfare loss is
$$\welfloss (\beta,\sigma_\eta)=(1-\beta)^2 \vareta +m^2 \beta^4 \vargamma - 2 (1-\beta) m \beta^2 \rho \sigma_\eta \sigma_\gamma.$$
Even though $\bstar$ depends on $\sigma_\eta$, the envelope theorem implies $\frac{\mathrm d \welfloss (\bstar,\sigma_\eta)}{\mathrm d \sigma_\eta} = \welfloss_{\sigma_\eta} (\bstar,\sigma_\eta)$. Evaluating this partial derivative, and substituting in $k\equiv m\sigma_\gamma/\sigma_\eta$ and 
\begin{equation}
\zeta(\beta)\equiv (1-\beta)- k \beta^2\rho,	\label{eq:dwelddeta_zeta}
\end{equation}
we compute
\begin{equation}
{\welfloss_{\sigma_\eta} (\bstar,\sigma_\eta)}=2 \sigma_\eta (1-\bstar) \zeta(\bstar).	\label{eq:dwelfdeta}
\end{equation}

The comparative statics in $\sigma_\eta$ are now directly implied by the following two claims.

\begin{claim}
\label{claim:dwelddeta_1}
If $\rho\geq 0$, then \eqref{eq:dwelfdeta} is positive.
\end{claim}

\begin{proof}
Assume $\rho\geq 0$. Then $\bstar \in (0,1)$ (\autoref{rem:bstar<1}), so it is sufficient to establish that $\zeta(\bstar)>0$. This inequality is immediate from \eqref{eq:dwelddeta_zeta} if $\rho=0$, so suppose $\rho>0$. From \eqref{eq:dwelddeta_zeta}, $\rho>0$ makes the function $\zeta$ a concave quadratic in $\beta$ that is positive at $\beta=0$ and whose only positive root is $\frac{-1+\sqrt{4 k \rho +1}}{2 k \rho }$. Since $\bstar$ is the first nonnegative point at which $L_\beta(\cdot,k,\rho)\geq 0$ (\autoref{lem:bstarsoln} part \ref{bstarfirstzero}), it is sufficient to show that 
$$L_\beta\left(\frac{-1+\sqrt{4 k \rho +1}}{2 k \rho },k,\rho\right)>0.$$
Evaluating \autoref{eq:Ldbeta} at $\beta=\frac{-1+\sqrt{4 k \rho +1}}{2 k \rho }$ and simplifying yields
\begin{equation}
\label{eq:dweldeta_step}
L_\beta\left(\frac{-1+\sqrt{4 k \rho +1}}{2 k \rho },k,\rho\right)=\frac{2 \left(1-\rho ^2\right) \left(k \rho  \left(\sqrt{4 k \rho +1}-3\right)+\sqrt{4 k \rho +1}-1\right)}{k \rho ^3},
\end{equation}
which is positive because $\rho>0$.\footnote{As $\rho \in (0,1)$, the sign of the right-hand side of 
\autoref{eq:dweldeta_step} is the same as that of its numerator's expression {$k \rho  \left(\sqrt{4 k \rho +1}-3\right)+\sqrt{4 k \rho +1}-1$}. This expression would equal $0$ were $k=0$, and its derivative with respect to $k$ is positive at any $k>0$.}
\end{proof}

\begin{claim}
Suppose $\rho< 0$. If $2k<-\rho$, then \eqref{eq:dwelfdeta} is negative; if $2k>-\rho$, then \eqref{eq:dwelfdeta} is positive.
\end{claim}

\begin{proof}
Assume $\rho<0$. We first recall from \autoref{rem:bstar<1} that $2k>-\rho \iff \bstar<1$ and $2k<-\rho\iff \bstar>1$.

So first suppose $\bstar \in (0,1)$. Then the sign of \eqref{eq:dwelfdeta} is the same as that of $\zeta(\bstar)$, which from \eqref{eq:dwelddeta_zeta} is evidently positive when $\bstar \in (0,1)$ and $\rho<0$.

Now suppose $\bstar>1$. Then the sign of \eqref{eq:dwelfdeta} is that of $-\zeta(\bstar)$. From \eqref{eq:dwelddeta_zeta}, $\rho<0$ makes the function $-\zeta$ a concave quadratic in $\beta$ that is negative at $\beta=0$ and either has no real roots or has roots $\frac{-1\pm \sqrt{4 k \rho +1}}{2 k \rho }$. 
Since $-\zeta$ is globally negative if it has no real roots, assume otherwise, i.e., $4k\rho +1\geq 0$. Since $\bstar$ is the first nonnegative point at which $L_\beta(\cdot,k,\rho)\geq 0$ (\autoref{lem:bstarsoln} part \ref{bstarfirstzero}), $-\zeta(\bstar)<0$ if
$$L_\beta\left(\frac{-1-\sqrt{4 k \rho +1}}{2 k \rho },k,\rho\right)>0.$$
Evaluating \autoref{eq:Ldbeta} at $\beta=\frac{-1-\sqrt{4 k \rho +1}}{2 k \rho }$ and simplifying yields
\begin{equation}
\label{eq:dweldeta_step2}
L_\beta\left(\frac{-1-\sqrt{4 k \rho +1}}{2 k \rho },k,\rho\right)=-\frac{2 \left(1-\rho ^2\right) \left((1+k \rho) \sqrt{4 k \rho +1}+1+3k \rho\right)}{k \rho ^3}
\end{equation}
which is positive because $\rho<0$.\footnote{As $\rho \in (-1,0)$, the right-hand side of 
\autoref{eq:dweldeta_step2} is the same as that of its numerator's expression $(1+k \rho) \sqrt{4 k \rho +1}+1+3k \rho$. This expression is positive because $\rho<0$ and $1+4k\rho\geq 0$ imply $1+k\rho>1+3k\rho>1+4k\rho\geq 0$.}
\end{proof}

\subsection{Proof of \autoref{prop:hybrid}}

Throughout this proof, we denote $\tau \equiv (1-\kappa)\eta+\kappa \gamma$, $\mu_{\tau} \equiv \E[\tau] = (1-\kappa) \mu_\eta + \kappa \mu_\gamma$, and $\sigma_\tau\equiv \sqrt{\Var(\tau)}=\sqrt{(1-\kappa)^2\vareta+\kappa^2\vargamma}$. 

We begin by proving the statement about fixed points. 

\begin{lemma}
\label{lem:hybridfp}
Assume $\rho=0$ and designer welfare \eqref{eq:hybridobj}. Among fixed points with positive sensitivity, there is a unique one.	
\end{lemma}

\begin{proof}
As explained in \autoref{sec:hybrid}, $\hat \beta(\beta)$ is the OLS regression coefficient $\Cov(x,\tau)/\Var(x)$. Since $\eta$ and $\gamma$ are uncorrelated,
\begin{align}
	\Cov(x,\tau)&=\Cov(\eta+m\beta\gamma,(1-\kappa)\eta+\kappa \gamma)=(1-\kappa)\sigma^2_\eta+m\beta\kappa\sigma^2_\gamma,\label{eq:cov_hybrid}\\
	\Var(x)&=\Var(\eta+m\beta\gamma)=\sigma^2_\eta+(m\beta)^2\sigma^2_\gamma.\label{eq:var_hybrid}
\end{align}
A fixed point $(\beta,\beta_0)$ satisfies $\hat \beta(\beta)=\beta$, which can be rewritten as the cubic equation
\begin{equation}
\label{eq:eqcubichybrid}
m^2\beta^3\vargamma+(\vareta-m\kappa\vargamma)\beta-(1-\kappa)\vareta=0.	
\end{equation}
The left-hand side of \eqref{eq:eqcubichybrid} is continuous, negative at $\beta=0$, and tends to $\infty$ as $\beta\to \infty$. There is a positive solution to \eqref{eq:eqcubichybrid} by the intermediate value theorem. The positive solution is unique because differentiation shows that the left-hand side of \eqref{eq:eqcubichybrid} is strictly convex on $\beta>0$.
\end{proof}
\begin{remark}
\label{rem:eqcubichybrid}
For subsequent use, note that when $\beq$ denotes the (unique) positive solution to \autoref{eq:eqcubichybrid}, the left-hand side of \eqref{eq:eqcubichybrid} is negative if and only if $\beta<\beq$, and positive if and only if $\beta>\beq$.	
\end{remark}

\paragraph{The designer's problem.} We next state explicitly and simplify the designer's problem. An optimal $(\bstar,\bstar_0)$ solves
 \[
 \min_{(\beta,\beta_0)\in \Reals^2} \E\left[ \left(\beta(\eta+m \beta \gamma) + \beta_0- \tau \right)^2 \right]
 \]
The first-order condition with respect to $\beta_0$ implies
 \[
 \beta^*_0 = \mu_\tau - \beta \mu_\eta - m \beta^2 \mu_\gamma.
 \]
 Substituting $\beta^*_0$ into the objective, the designer chooses $\beta$ to minimize
 \begin{align*}
%    \E\left[ \bigl(\beta(\eta+m \beta \gamma) +\mu_\tau-\beta \mu_\eta-m \beta^2 \mu_\gamma-\tau\bigr)^2 \right] \\
   & \E\left[ \left(\beta(\eta-\mu_\eta) +m \beta^2 (\gamma-\mu_\gamma) -(\tau-\mu_\tau) \right)^2 \right]\\
%   & =\E\left[ \beta^2(\eta-\mu_\eta)^2 +(m \beta^2)^2 (\gamma-\mu_\gamma)^2+(\tau-\mu_\tau)^2\right]\\
%   		& \quad +2\E\left[(\eta-\mu_\eta)(\gamma-\mu_\gamma)-(\eta-\mu_\eta)(\tau-\mu_\tau)-(\gamma-\mu_\gamma)(\tau-\mu_\tau) \right] \\
%   & =\beta^2\sigma_\eta^2 +(m \beta^2)^2 \sigma_\gamma^2+(1-\kappa)^2\sigma^2_\eta+\kappa^2 \sigma^2_\gamma\\
%   		& \quad -2\E\left[(\eta-\mu_\eta)(\tau-\mu_\tau)+(\gamma-\mu_\gamma)(\tau-\mu_\tau) \right] \qquad{\small \text{since $\eta$ and $\gamma$ are uncorrelated}} \\
%   & =\beta^2\sigma_\eta^2 +(m \beta^2)^2 \sigma_\gamma^2+(1-\kappa)^2\sigma^2_\eta+\kappa^2 \sigma^2_\gamma\\
 %  		& \quad -2(1-\kappa)\sigma^2_\eta-2\kappa\sigma^2_\gamma \qquad{\small \text{since $\eta$ and $\gamma$ are uncorrelated}} \\
   = & (\beta-(1-\kappa))^2 \sigma_\eta^2 + (m \beta^2 - \kappa)^2 \sigma_\gamma^2  \\
   = & \sigma_\eta^2 \left[ (\beta-(1-\kappa))^2 + (m \beta^2 - \kappa)^2 l^2 \right],
 \end{align*}
where the first equality uses $\E\left[(\eta-\mu_\eta)(\gamma-\mu_\gamma)\right]=0$, $\E\left[(\eta-\mu_\eta)(\tau-\mu_\tau)\right]=(1-\kappa)\sigma^2_\eta$, and $\E\left[(\gamma-\mu_\gamma)(\tau-\mu_\tau)\right]=\kappa \sigma^2_\gamma$, and in the second equality
    \[
    l \equiv {\sigma_\gamma}/{\sigma_\eta}>0.
    \]
Equivalently, for $l>0$, the designer minimizes
$$L(\beta) \equiv (\beta-(1-\kappa))^2 + (m \beta^2 - \kappa)^2 l^2 .$$

\paragraph{Unique optimum.} We establish that there is a unique minimizer of $L(\beta)$, $\bstar>0$.

\underline{Step 1:} Any minimizer of $L(\beta)$ is positive.

Algebra shows that $L(\beta)-L(-\beta)=-4\beta(1-\kappa)$, any hence any optimum is nonnegative. Moreover $L'(\beta)= 2(\beta-(1-\kappa))+4 \beta  l^2 m \left(\beta ^2 m-\kappa \right)$, which 
is negative at $\beta=0$.

\underline{Step 2:} $L'(\beta)=0$ has a unique solution on $\beta>0$. Hence, $L(\beta)$ has a unique local minimizer on $\beta>0$ and this is the unique global minimizer (over $\beta\in \Reals$).

We compute $L'''(\beta)=24l^2m^2 \beta$, which is positive for $\beta>0$. That is, $L'(\beta)$ is strictly convex on $\beta>0$. Inspecting $L'$ computed in Step 1, we see that $L'(0)<0$ and $L'(\beta)\to \infty$ as $\beta\to \infty$. Hence, $L'(\beta)=0$ has a unique solution on $\beta>0$.

\paragraph{Comparison of $\bstar$ and $\beq>0$.} We claim that at the fixed point $(\beq,\beq_0)$ with $\beq>0$:
\begin{align*}
\kappa<\bar \kappa &\implies L'(\beq)>0,\\
\kappa = \bar \kappa &\implies L'(\beq)=0,\\
\kappa>\bar \kappa &\implies L'(\beq)<0.
\end{align*}
Given Step 2 above, this claim proves \autoref{prop:hybrid}.

To prove the claim, we begin with a decomposition analogous to our main analysis (see, in particular, the \hyperref[proof:lossderivbeq]{proof of} \autoref{lem:lossderivbeq}):
$$\vareta L(\beta) = \underbrace{\sigma^2_\tau\left(1-\Corr(x,\tau)^2\right)}_{\text{Info loss from linearly estimating $\tau$ using $x$}}+\underbrace{(\beta-\hat \beta(\beta))^2\var(x)}_{\text{Misallocation loss given linear estimation}},$$
where $\Corr(x,\tau)=\hat\beta(\beta)\sqrt{\Var(x)}/\sigma_\tau$, 
$\hat\beta(\beta)=\Cov(x,\tau)/\Var(x)$ and the formulas for $\Cov(x,\tau)$ and $\Var(x)$ are given in \autoref{eq:cov_hybrid} and \autoref{eq:var_hybrid}.

Differentiating,
\begin{align*}
\vareta L'(\beta) &= \overbrace{- 2 \sigma^2_\tau \Corr(x,\tau) \frac{\mathrm d }{\mathrm d \beta} \Corr(x,\tau)}^{\text{Marginal change in info loss}} \\ &\qquad + \underbrace{\left( -2 (\beta - \hat \beta(\beta)) \hat \beta'(\beta) \var(x) +   (\beta - \hat \beta(\beta))^2 \frac{\mathrm d}{\mathrm d \beta} \var(x)\right)}_{\text{Marginal change in misallocation loss}}.
\end{align*}
Let us evaluate this expression at $\beta=\beq>0$. 
Since $\beq=\hat \beta(\beq)$, the marginal change in misallocation loss is evidently zero. % (intuitively because the misallocation loss is minimized at $\beta=\beq$). 
Thus,
\begin{align*}
%\label{eq:welfloss'1}
\sign L'(\beq) &= \left.- \left(\sign \Corr(x,\tau)\right) \left( \sign \frac{\mathrm d }{\mathrm d \beta} \Corr(x,\tau)\right)\right|_{\beta=\beq}= \left.- \sign \frac{\mathrm d }{\mathrm d \beta} \Corr(x,\tau)\right|_{\beta=\beq},
\end{align*}
where the second equality is because $\sign \left.\Corr(x,\tau)\right|_{\beta=\beq}=\sign \hat \beta(\beq)=\sign \beq>0$.
Dividing \autoref{eq:cov_hybrid} by $\sigma_\tau\sqrt{\Var(x)}$ and differentiating,
$$\frac{\mathrm d }{\mathrm d \beta} \Corr(x,\tau)=\frac{m \vareta \vargamma (\kappa -\beta  (1-\kappa ) m)}{\sigma_\tau \Var(x)^{3/2}}.$$
Thus,
\begin{align*}
\sign L'(\beq)&=\sign\left(\beq(1-\kappa)m-\kappa\right)=-\sign \phi\left(\frac{\kappa}{m(1-\kappa)}\right),
\end{align*}
where $\phi(\beta)$ is the left-hand side of \autoref{eq:eqcubichybrid}; see \autoref{rem:eqcubichybrid}.
Some algebra shows that $$\phi\left(\frac{\kappa}{m(1-\kappa)}\right)=-\frac{\left( m (1-\kappa)^2-\kappa \right) \left((1-\kappa)^2 \vareta+\kappa ^2 \vargamma\right)}{(1-\kappa)^3 m},$$
and thus
\begin{align*}
\sign L'(\beq)&=\sign \left( m (1-\kappa)^2-\kappa \right).
\end{align*}
The argument of the sign function of this equality's right-hand side is a quadratic in $\kappa$ that is positive and decreasing at $\kappa=0$ and negative at $\kappa=1$. The quadratic's only root in $(0,1)$ is $(1 + 2 m - \sqrt{1 + 4 m})/(2 m)$, i.e., $\bar \kappa$. The claim follows.

%\newpage

\vspace{1in}

\bibliographystyle{ecta}
\bibliography{Frankel-Kartik}

%\end{spacing}

\newpage

%%%%%%%%%%%%%%%%%%%%%%%%%%%%%%

\noindent {\LARGE \textbf{Supplementary (Online) Appendices}}

\section{Monotone Comparative Statics}
The following fact on monotone comparative statics is used in the proof of \autoref{prop:main} and in the proof of \autoref{prop:compstats}. Although it is well known, we include a proof.
\begin{fact}
\label{MCS}
Let $T\subseteq \Reals$, $Z\subseteq \Reals$ be open, and $f:Z  \times T\to \Reals$ be continuously differentiable in $z$ with for all $t\in T$, $\argmin_{z\in Z} f(z,t)\neq \emptyset$. Define $M(t) \equiv \argmin_{z\in Z} f(z,t)$. For any $\bar t\in T$ and $\underline t \in T$ with $\bar t> \underline t$, it holds that:\reducespace
\begin{enumerate}
\item If $f_{z}(z,\bar t)>f_z(z,\underline t)$ for all $z\in Z$, then for any $\bar m \in M(\bar 
t)$ and any $\underline m \in M(\underline t)$ it holds that $\bar m<\underline m$.

%\begin{comment}
\underline{Proof}: For any $\hat z>\underline m$,
$$f(\hat z,\o t)-f(\underline m,\o t)=\int_{\underline m}^{\hat z} f_z(z,\o t)\mathrm d z>\int_{\u m}^{\hat z} f_z(z,\underline t)\mathrm d z = f(\hat z,\underline t)-f(\underline m,\underline t)\geq 0.$$
Hence $\overline m \leq \underline m$. The inequality must be strict because otherwise the first-order conditions yield $0=f_z(\o m,\o t)=f_z(\underline m,\overline t)>f_z(\underline m,\underline t)=0$.
%\end{comment}
\item If $f_{z}(z,\bar t)<f_z(z,\underline t)$ for all $z\in Z$, then for any $\bar m \in M(\bar 
t)$ and any $\underline m \in M(\underline t)$ it holds that $\bar m>\underline m$.
(We omit a proof, as it is analogous to that above.) \hfill \textsquare
\begin{comment}
\smallskip
\underline{Proof}: For any $\hat z>M(\overline t)$,
$$f(\hat z,\u t)-f(M(\o t),\u t)=\int_{M(\o t)}^{\hat z} f_z(z,\u t)\mathrm d z>\int_{M(\o t)}^{\hat z} f_z(z,\o t)\mathrm d z = f(\hat z,\o t)-f(M(\o t),\o t)\geq 0.$$
Hence $M(\u t)\leq M(\o t)$. The inequality must be strict because otherwise the first-order conditions yield $0=f_z(M(\o t),\overline t)=f_z(M(\underline t),\overline t)<f_z(M(\underline t),\underline t)=0$.
\end{comment}
\end{enumerate}
\end{fact}

\section{Alternative Model of Information Loss}
\label{app:binary}

%The fundamental logic underlying our paper's main result is simply that ``flattening'' the allocation rule from any fixed point yields a first-order improvement in information while only a second-order loss from misallocation. 
Our paper finds that a designer improves information, and thereby allocation accuracy, by flattening a fixed point rule. We developed this point in what we believe is a canonical model of information loss from manipulation, one used in a number of other papers. But we think the point applies more broadly, including in other models of information loss.  For instance, even a model with a one-dimensional type (such as the model in this paper with no heterogeneity on the gaming ability $\gamma$) can lead to information loss when there is a bounded action space and strong manipulation incentives. The reason is ``pooling at the top''. We establish below a version of our main result for a simple model in this vein.

Let the agent take action $x \in \{0,1\}$ with natural action $\eta \in \{0,1\}$.  The agent's type $\eta$ is her private information, drawn with ex-ante probability $\pi\in (0,1)$ that $\eta=1$. After observing $x$, the designer chooses allocation $y\in \Reals$ with payoff $-(y-\eta)^2$. We assume, for simplicity, that the agent of type $\eta=1$ must choose $x=1$.\footnote{Our main point goes through so long as action $x=1$ is no more costly than $x=0$ for type $\eta=1$, as this will ensure it is optimal for type $\eta=1$ to choose $x=1$.} The payoff for type $\eta=0$ is $y - cx$, where $c>0$ is a commonly known parameter. To streamline the analysis, we assume $c \in (0,\pi)$.

A pure allocation rule or policy is $Y:\{0,1\}\to \Reals$. Due to the designer's quadratic loss payoff, it is without loss to focus on pure policies. Given a policy $Y$, let $\Delta\equiv Y(1)-Y(0)$ be the difference in allocations across the two actions of the agent. We focus, without loss, on policies with $\Delta\ge 0$. A policy with a smaller $\Delta$ is a ``flatter'' policy, i.e., it is less sensitive to the agent's action. The naive policy $Y^\mathrm{n}$ sets $Y^\mathrm{n}(1)=1$ and $Y^\mathrm{n}(0)=0$, corresponding to a naive allocation difference of $\Delta^n=1$. Let $\Delta^{\mathrm{fp}}$ and $\Delta^*$ denote the corresponding differences from fixed point and commitment policies.

\subsection{Naive Policy}

Take any policy with $\Delta=1$. Since we assume $c<\pi<1$, even the agent with $\eta=0$ will then choose $x=1$. So welfare---the designer's ex-ante expected payoff---from the naive policy is 
$$-\pi(0-0)^2 -(1-\pi)(1-0)^2=-(1-\pi).$$
%The designer's optimal policy with $Y(0)$ and $Y(1)$ such that $\Delta=1$ -- then the designer would want to set $Y(1) =\pi$ (and $Y(0) = \pi-1$), in which case the payoff is $ -\pi (1-\pi)^2- (1-\pi) \pi^2 = -\pi (1-\pi)$.

\subsection{Fixed Point}
At a Bayesian Nash equilibrium (of either the simultaneous move game, or when the agent moves first), $Y(x)=\mathbb{E}[\eta|x]$ for any $x$ on the equilibrium path. If $x=0$ is on the equilibrium path, $Y(0)=0$ because type $\eta=1$ does not play $x=0$. %We may assume without loss that $Y(0)=0$ in any equilibrium.

There is a fully-pooling equilibrium with both types playing $x=1$: the designer plays $Y(1)=\pi$ and $Y(0)=0$, and it is optimal for type $\eta=0$ to play $x=1$ because $c<\pi$. The corresponding welfare is
$$-\pi(\pi-1)^2-(1-\pi)(\pi-0)^2=-\pi(1-\pi).$$

There is no equilibrium in which the agent of type $\eta=0$ puts positive probability on action $x=0$, because that would imply $Y(1)>\pi$ and $Y(1)=0$, against which the agent's unique best response is to play $x=1$.

Therefore, we have identified the (essentially unique, up to the off-path allocation following $x=0$) fixed point policy: $Y^\mathrm{fp}(1)=\pi$, $Y^\mathrm{fp}(0)=0$, and therefore $\Delta^{\mathrm{fp}}=\pi$. The agent pools on $x=1$, and welfare is $-\pi(1-\pi)$.\footnote{\label{fn:offpath}The choice of $Y^\mathrm{fp}(0)=0$ can be justified from the perspective of the agent ``trembling''. In particular, in the signaling game where the agent moves before the designer, any sequential equilibrium \citep{KW82} has $Y(0)=0$, as only type $\eta=0$ can play $x=0$. But note that no matter how $Y(0)$ is specified, it must hold in a fixed point that $\Delta \leq  c$; otherwise the agent will not pool at $x=1$.} This welfare is larger than that of the naive policy.

\subsection{Commitment}
Now suppose the designer's commits to a policy before the agent moves. From the earlier analysis, if $\Delta >c$ the agent will pool at $x=1$ and so an optimal such policy is the fixed point policy $Y^\mathrm{fp}$.
For any $\Delta<c$, there is full separation: the agent's best response is $x=\eta$. Indeed, full separation is also a best response for the agent when $\Delta=c$. Given that the designer wants to match the agent's type, it follows that the optimal way to induce full separation  is to set $\Delta=c$ (or $\Delta = c^-$), i.e., have $Y^*(1)=Y^*(0)+c$. 

At such an optimum, quadratic loss utility implies that the designer sets an average action of $(1-\pi) Y^*(0) + \pi Y^*(1)$ equal to $\E[\eta]=\pi$. Plugging in $Y^*(1) = Y^*(0) + c$ yields
\[(1-\pi) Y^*(0) + \pi (Y^*(0)+c) = \pi,\]
and hence the solution
\[Y^*(0) = \pi (1-c), \quad Y^*(1) = \pi (1-c) + c.\]

The corresponding welfare is 
\begin{align*}
- (1-\pi)( \pi(1-c)-0)^2 - \pi (\pi(1-c)+c-1)^2 &=  - (1-c)^2 (1-\pi)\pi.
\end{align*}

This welfare is larger than that under the fixed point. Moreover, the optimal policy has $\Delta^*=c$ while the fixed point has $\Delta^\mathrm{fp}=\pi$ and the naive policy has $\Delta^\mathrm{n}=1$. Thus the optimal policy is flatter than the fixed point, which in turn is flatter than the naive policy: %\footnote{The first inequality below be an equality if one takes the off-path path policy in the fixed point to be $Y^\mathrm{fp}(0)=\pi-c$ (see ~\autoref{fn:offpath}). But one can then view the commitment solution as having $\Delta^*=c^-$ to recover the inequality.}
$$\Delta^* < \Delta^{\mathrm{fp}} < \Delta^\mathrm{n}.$$

Note that the designer obtains no benefit from reducing $\Delta$ from $\Delta^{\mathrm{fp}}=\pi$ until reaching $\Delta^*=c$; this is an artifact of the assumption that there is no heterogeneity in the manipulation cost $c$. In a model with such heterogeneity, there would be a more continuous benefit of reducing $\Delta$ from the fixed point.
\end{document}